\documentclass[nonacm,sigconf,10pt]{acmart}
\usepackage{tikz}
\usepackage{booktabs}

\usepackage[]{hyperref}
\usepackage{xspace}
\usepackage{subcaption}
\usepackage{amsmath}
\usepackage{amsthm}
\usepackage{amsfonts}
\usepackage{enumitem}
\usepackage[ruled,linesnumbered]{algorithm2e}
\usepackage{bbm}
\newtheorem{assumption}{Assumption}

\newtheorem{theorembody}{Theorem}
\newtheorem{theorem}{Theorem}
\newtheorem{proposition}{Proposition}

\def\Snospace~{\S{}}

\usepackage{xcolor} 

\newcommand{\sys}{Lazarus\xspace}
\newcommand{\VSPACE}[1]{\vspace{#1}}
\renewcommand{\VSPACE}[1]{}

\begin{document}
\date{}

\title[]{\sys: Resilient and Elastic Training of Mixture-of-Experts Models}

\author{
\rm{Yongji Wu$^{\text{1}, *}$\enskip
    Wenjie Qu$^{\text{2}, *}$ \enskip
    Xueshen Liu$^{\text{3}, *}$ \enskip
    Tianyang Tao$^{\text{2}}$ \enskip
    Yifan Qiao$^{\text{1}}$ \enskip \\
\rm{Zhuang Wang$^{\text{4}}$ \enskip
    Wei Bai$^{\text{5}}$ \enskip} 
    Yuan Tian$^{\text{6}}$ \enskip
    Jiaheng Zhang$^{\text{2}}$ \enskip
    Z. Morley Mao$^{\text{3}}$ \enskip
    Matthew Lentz$^{\text{7}}$ \enskip
    Danyang Zhuo$^{\text{7}}$ \enskip
    Ion Stoica$^{\text{1}}$ \enskip}\\ 
{$^{\text{1}}$UC Berkeley\enskip$^{\text{2}}$NUS\enskip$^{\text{3}}$UMich$\enskip^{\text{4}}$AWS\enskip$^{\text{5}}$NVIDIA\enskip$^{\text{6}}$UCLA\enskip$^{\text{7}}$Duke}
}
\renewcommand{\shortauthors}{}

\begin{abstract}
Sparsely-activated Mixture-of-Experts (MoE) architecture has increasingly been adopted to further scale large language models (LLMs). However, frequent failures still pose significant challenges as training scales. The cost of even a single failure is significant, as all GPUs need to idle wait until the failure is resolved, potentially losing considerable training progress as training has to restart from checkpoints. This problem is exacerbated by the growing use of spot instances on public clouds for model training, which despite offering substantial cost savings, introduce frequent preemptions—essentially failures that  regularly occur throughout the training process. Existing solutions for efficient fault-tolerant training either lack elasticity or rely on building resiliency into pipeline parallelism, which cannot be applied to MoE models due to the expert parallelism strategy adopted by the MoE architecture. 

We present \sys, a system for resilient and elastic training of MoE models. \sys adaptively allocates expert replicas to address the inherent imbalance in expert workload and speeds up training, while a provably optimal expert placement algorithm is developed to maximize the probability of recovery upon failures. Through adaptive expert placement and a flexible token dispatcher, \sys can also fully utilize all available nodes after failures, leaving no GPU idle. Our evaluation shows that \sys outperforms existing MoE training systems by up to 5.7x under frequent node failures and 3.4x on a real spot instance trace. 
\end{abstract}

\maketitle

{\let\thefootnote\relax\footnote{
$^*$Yongji Wu, Wenjie Qu and Xueshen Liu are co-first authors of this work.
}}

\section{Introduction}

The advent of large language models (LLMs) has demonstrated ever-increasing capabilities with the rapid growth in both model sizes and training datasets. Recently, the sparsely-activated Mixture-of-Experts (MoE) models have been increasingly adopted by the community to further scale model parameters~\cite{liu2024deepseek,llama4,qwen3,jiang2024mixtral}. Training state-of-the-art MoE models is becoming resource-intensive. For instance, it takes over 32K H100 GPUs to train the 2T Llama 4 model~\cite{llama4}.

The likelihood and frequency of failures significantly increase as the scale and duration of training increase. Meta projects that the mean time to failure (MTTF) is as little as 14 minutes for a cluster with 128K GPUs~\cite{kokolis2025revisiting}. Even a single failure is costly, as all GPUs are idle until the failure is resolved and failed nodes are replaced. It is reported that failures can slow the training progress by up to 43\%~\cite{maeng2021understanding}. In addition, most cloud providers offer preemptible (spot) instances that can be leveraged for training LLMs with minimized monetary cost~\cite{duan2024parcae,thorpe2023bamboo}, as they offer cost savings of up to 90\% compared to on-demand instances. Preemptions, which are essentially failures, can happen as frequently as every 5\textasciitilde10 minutes~\cite{thorpe2023bamboo}.

Existing solutions for LLM training with quick failure recovery can be categorized into two classes: checkpointing optimizations or pipeline-parallelism based elastic training. The first line of work~\cite{wang2023gemini,wang2023reliable,cai2025moc} reduces checkpointing overhead by either using CPU memory of neighboring nodes to periodically checkpoint model states, or relying on stale states which compromises correctness~\cite{cai2025moc}.
They also lack elasticity and have to wait for replacement nodes of the failed ones to recover from failure and continue training, which may not be available for hours to days until failed nodes are repaired~\cite{he2023unicron}. Especially for training on spot instances, such new node availability cannot be taken for granted.

The second line of works builds resiliency and elasticity into pipeline parallelism by taking advantage of its configurability in stages-nodes mapping~\cite{thorpe2023bamboo,jang2023oobleck,duan2024parcae}. In particular, they can continue training upon failures without requiring additional nodes. However, these approaches do not apply to MoE models, as the distributed training of MoE models depends on a different parallelism strategy: expert parallelism (EP)~\cite{lepikhin2020gshard}.
EP distributes experts across multiple GPUs (and nodes) and uses all-to-all communication to dispatch input tokens to GPUs with corresponding experts.

In this paper, we present \sys, a system for resilient and elastic training of MoE models. \sys achieves high-throughput training accompanied by a high failure recovery probability without restarting from checkpoints. Upon failures, \sys quickly reconfigures the training job and utilizes all remaining GPUs (regardless of how many nodes fail).

Our insight is that adaptively adjusting the number of replicas (GPUs) assigned to each expert and their placement enables elastic training while improving resiliency against failures.
Due to the dynamic nature of its architecture, MoE models suffer from dynamic and imbalanced workload~\cite{hwang2023tutel,zhai2023smartmoe,nie2023flexmoe}. Tokens are routed to experts based on the decisions of trainable gate networks. Some experts have more tokens routed to than others. Traditional EP partitions experts into equal-sized chunks, and each is assigned to the same number of GPUs. In contrast, \sys allocates more replicas to popular experts and flexibly assigns them using all available GPUs. Such flexible expert allocation not only results in performance boosts but also leads to better elasticity. As long as a single replica for each expert remains available, training can continue to progress with all remaining nodes utilized; traditional EP requires using a multiple of EP size GPUs, which can induce significant performance degradation even for minor failures.

There are three key challenges \sys must address. 
First, we need an expert allocation and placement algorithm that takes account of the imbalanced workload, to speed up expert computation while ensuring a high probability of successful recovery.
Second, with our asymmetrical expert placements in the cluster, how do we efficiently dispatch tokens to GPUs with corresponding experts and balance their loads? 
Third, how do we quickly re-instantiate lost expert replicas and efficiently migrate the cluster to a new placement plan in response to failures?

\begin{figure}
\centering
\includegraphics[width=0.95\linewidth]{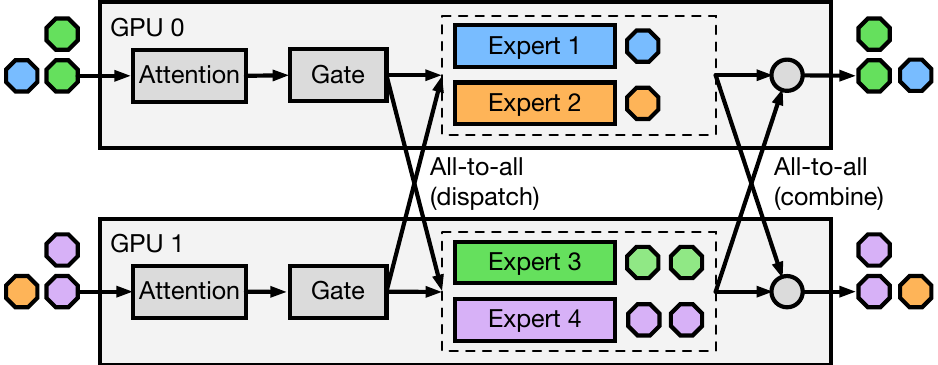}
\caption{MoE architecture utilizes expert parallelism for distributed training, yet it also suffers from imbalanced workload due to the dynamic nature of gate networks.}
\label{fig:intro}
\VSPACE{-4mm}
\end{figure}

To address these challenges, we propose a strategy for allocating expert replicas based on the load distribution, while maintaining a fault-tolerant threshold to guarantee failure recovery when a small number of nodes fail. We design a provably optimal algorithm for placing these replicas to maximize the recovery probabilities under arbitrary node failures. We develop a CUDA kernel that dispatches tokens in parallel with a flexible all-to-all that minimizes inter-GPU communication. During migration, \sys utilizes a greedy strategy to reduce state transfers for efficient reconfiguration.

We implement \sys in PyTorch. We evaluate \sys across MoE models of different scales with both controlled failures and spot instance traces. Our results show that \sys outperforms checkpointing-based DeepSpeed MoE~\cite{rajbhandari2022deepspeed}, a widely adopted system for training MoE models, by up to 2.3x under infrequent failures (40~mins MTBF) and 5.7x under a high failure frequency (5~mins MTBF), while our evaluation on a real spot instance trace demonstrates a performance improvement of 3.4x.

In this paper, we make the following contributions:
\begin{itemize}[leftmargin=*]
\item To the best of our knowledge, \sys is the first system for resilient and elastic training of MoE models that enables both quick recovery from failures and full utilization of all available  (remaining) GPUs.

\item We design a provably optimal algorithm for determining expert placement that maximizes recovery probability in response to uniformly random node failures.

\item We implement and evaluate \sys with MoE models of different scales under a variety of scenarios.

\end{itemize}

\section{Background and Motivation}
\subsection{MoE Models and Expert Parallelism}

\begin{figure}
    \begin{subfigure}{0.48\linewidth}
        \centering
        \includegraphics[width=\linewidth]{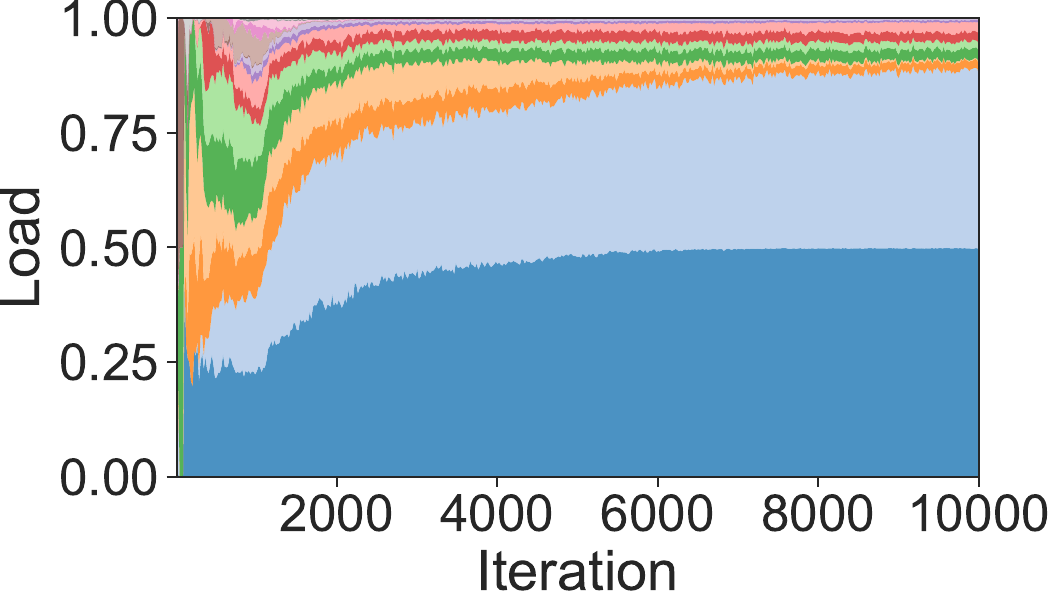}
        \caption{Layer 1}
    \end{subfigure} \hfil
    \begin{subfigure}{0.48\linewidth}
    \centering
      \includegraphics[width=\linewidth]{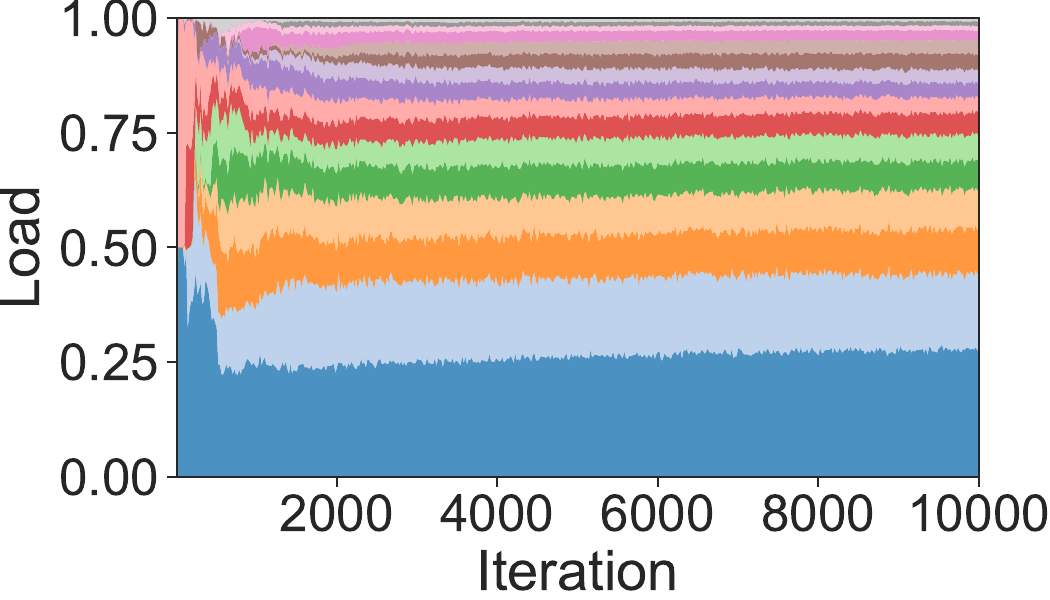}
        \caption{Layer 8}
    \end{subfigure}
\caption{Expert loads on a 16 experts model (GPT-L in \autoref{sec:exp-setup}). The distribution varies during training and across layers.}
\label{fig:background_expert_load}
\VSPACE{-4mm}
\end{figure}

Mixture-of-Experts architecture has been recently applied to scale LLMs due to its high cost-efficiency, which replaces the dense feed-forward network (FFN) in a transformer block. MoE employs multiple parallel FFNs called experts. In each MoE layer, a trainable gate network routes
each token to only the top-$k$ experts. As experts are sparsely activated, MoE enables scaling model parameters without an increase of the per-token computational cost. 

As the size of an MoE model is dominated by the weights of the experts, expert parallelism (EP)~\cite{lepikhin2020gshard} has been proposed and has become the de facto approach to train large-scale MoE models. In expert parallel training, the experts of each layer are split into equal-sized chunks and allocated across multiple GPUs similar to tensor parallelism, while the input samples are distributed along the batch dimension similar to data parallelism. The number of GPUs required to split the experts is called the EP size and such a set of GPUs forms an EP group. For instance, in \autoref{fig:intro}, there are 4 experts and each GPU accommodates 2 experts, therefore it has a EP size of 2. EP can be used in conjunction with other types of parallelism like data and tensor parallelism. 

As each GPU in an EP group only holds a subset of experts, all-to-all communication is used to dispatch the input tokens to the GPUs with corresponding experts that the gate network routes to. The computation of the experts are performed on the owning GPUs and the results are sent back to the original GPUs with a second all-to-all (combine).

The most distinctive feature of expert parallelism is the dynamic nature of gate networks. The distribution of tokens routed to each expert can be highly unbalanced depends on the input data. We plot the evolution of expert loads from a training trace~\cite{zhai2023smartmoe} in \autoref{fig:background_expert_load}. We observe that the load of experts is highly skewed, with up to 87\% tokens routed to 2 most popular experts. The load distribution also varies at different layers and training iterations.

The skewed expert loads in MoE training directly translates to imbalance in expert computation. GPUs holding more popular experts takes much longer time to compute due to large amount of tokens dispatched to them, while other GPUs are idling. Previous works~\cite{hwang2023tutel,nie2023flexmoe,zhai2023smartmoe,he2022fastermoe} addresses this challenge by dynamically adjusting parallelism strategies on a cluster with a fixed number of GPUs. They do not apply in an elastic environment with changing device membership.

In addition to the problem of imbalanced workload, traditional EP also utilizes a multiple of EP size GPUs, which may leave some of GPUs idle upon a failure.
The waste of GPUs only grows with increasing number of experts, as more GPUs are needed for a single EP group, i.e., larger EP size. 

\VSPACE{-3mm}
\subsection{Fault-Tolerant and Elastic Training}

A growing research effort has been made in resilient training in recent years, due to the fact that both the frequencies and costs of failures increase as the scale and duration of training increase. It is reported during the two-months training of OPT 175B, around 100+ failures were encountered~\cite{zhang2022opt}, wasting over 178,000 GPU hours. The cost of even one failure is significant, as all the GPUs must wait idle until the failure is resolved and failed nodes are repaired, which could take hours to days depending on the nature of failures~\cite{he2023unicron}. To minimize the GPU idling and the resulting economic loss, a training system must be designed with resiliency in terms of it can quickly recover from failures, and elasticity in terms that it can efficiently utilize currently available GPU resources to continue training. Such systems also enable one to leverage preemptible instances on public clouds to train LLMs with  significant cost savings~\cite{thorpe2023bamboo,duan2024parcae}.

Existing training solutions with quick failure recovery capability can be divided into two categories: checkpointing optimizations and elastic training using pipeline parallelism.  Checkpointing based solutions focus on reducing the overhead in both saving checkpoints and restarting~\cite{wang2023reliable,wang2023gemini,wan2024bytecheckpoint,cai2025moc,gandhi2024moetion}. In particular, in-memory based checkpointing~\cite{wang2023reliable,wang2023gemini} has been proposed to store model states in the CPU memory of other nodes in addition to persistent storage, while MoC-System~\cite{cai2025moc}, an MoE specific checkpointing solution, compromises correctness by using stale states. However, they lack elasticity as they have to wait until replacements of failed nodes are available to resume training.

To support both elastic and fault tolerant training without the overhead of checkpointing and restarting, recent attempts~\cite{thorpe2023bamboo,jang2023oobleck,duan2024parcae} have been made in building resiliency into pipeline parallelism due to its configurability. 
However, they fail to apply to MoE models. 
As the model states of a single MoE layer can exceed the GPU memory capacity, they are generally trained in conjunction with expert parallelism, requiring resiliency for expert states distributed across GPUs.

In summary, existing systems for fault-tolerant and elastic training fail to adapt to MoE models. \sys targets MoE training, utilizing adaptive expert allocation and placement to address expert parallelism's inelastic nature while handling the imbalanced expert load distribution caused by the dynamic gate networks.
\VSPACE{-2mm}
\section{System Overview}
\begin{figure}
\centering
\includegraphics[width=0.85\linewidth]{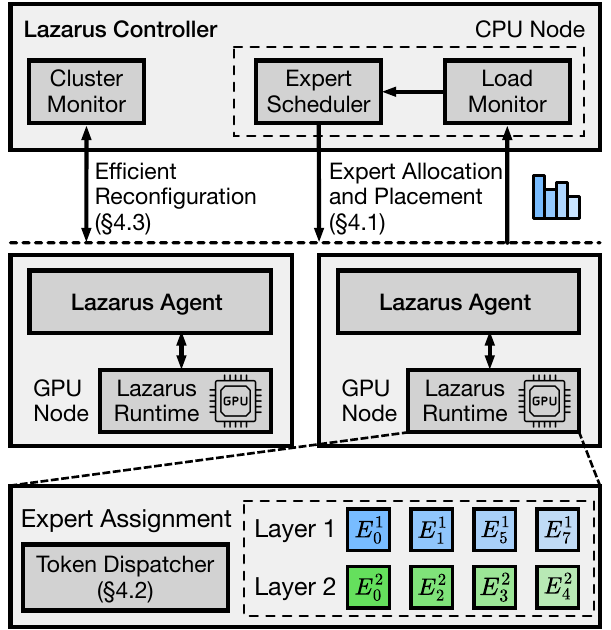}
\caption{System architecture of \sys. }
\label{fig:overview}
\VSPACE{-4mm}
\end{figure}

\sys is a resilient and elastic system for training MoE models. \sys speeds-up training by adaptively allocating expert replicas based on the dynamic expert load distribution using all available GPUs, while our fault-tolerant expert placement strategy maximizes \sys's recovery probability even under simultaneous failures of multiple nodes.

The architecture of \sys is shown in \autoref{fig:overview}. \sys consists of three main components: a centralized controller that manages a GPU cluster, an agent process on each GPU node that spins up worker processes with \sys runtime. The controller runs persistently on a (CPU-only) node and it communicates  with each \sys agent, monitors the cluster and detects node failures and replenishment. A scheduler in the controller allocates expert replicas and computes a fault-tolerant placement plan for all GPU nodes that maximizes the recovery probability (\autoref{sec:expert-placement}). The placement is sent to each \sys agent to configure the workers. Based on the placement plan, \sys runtime fills up each layer with corresponding experts assigned to it. Unlike vanilla expert parallelism where all experts are equally replicated, \sys assigns more replicas and more GPUs to the heavily loaded experts. As the expert placement becomes asymmetric, \sys runtime also contains a CUDA kernel based dispatcher (\autoref{sec:token-dispatch}) to efficiently dispatch tokens to GPUs with corresponding experts and balance their loads. 

Upon detection of failures, the controller recomputes an expert placement plan using all remaining nodes and minimizes the number of replicas migrated. Once \sys runtime receives the new plan relayed by \sys agent, it dynamically reconfigures the parallelism setups and retrieves missing model states from other nodes (\autoref{sec:effcient-reconfig}). To handle dynamics in workloads, \sys agent also periodically collects the expert load distribution (routing history of gate networks) from \sys runtime. The load distribution is communicated to the load monitor on the controller, which then rebalances the expert allocation and placement.

\VSPACE{-2mm}
\section{Design}
\subsection{Adaptive Expert Allocation and Placement}
\label{sec:expert-placement}
\sys considers that each GPU can hold a certain number of replicas limited by its GPU memory, similar to traditional EP.
Through assigning more replicas to popular experts, \sys can speed up training by giving them more computation resources.
Note that we allow multiple replicas of the same expert assigned to a single GPU, which indicates more tokens (of the specific expert) can be processed by that GPU compared to assigning a single replica. 

Yet, there is an inherent trade-off between speeding up computation and fault resiliency. On the one hand, if a less popular expert is assigned with only a single replica, then as long as the GPU (node) holding that replica fails, \sys cannot recover due to the loss of expert state.
On the other hand, a balanced allocation of replicas improves fault resiliency; however, it degenerates to traditional EP and defeats the goal of addressing expert load imbalance.

Moreover, given an expert replica allocation, the placement of these replicas determines the probability of failure recovery. For instance, if all replicas of an expert are all placed on GPUs in a single node, the loss of that node would lead to an unrecoverable failure. Hence, when allocating and placing expert replicas, \sys should not only consider the imbalanced workload to speed up computation, but also take account of the impact on fault tolerance.

\begin{figure}
\centering
\includegraphics[width=.85\linewidth]{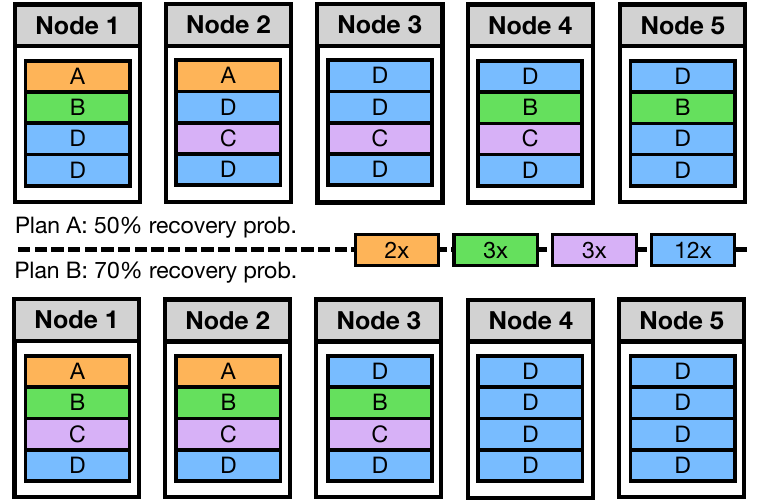}
\caption{Fault resiliency depends on how expert replicas are placed. With the same replica allocation of 4 experts and 4 replica slots per node, placement plan A and B differ in recovery probability under 3 node failures.}
\label{fig:recover_exp1}
\VSPACE{-4mm}
\end{figure}

We divide the problem into two phases and separately consider allocation and placement. In the first phase, we design an expert allocation strategy that balances between the workload's expert distribution and fault tolerance. In the second phase, we design an expert placement algorithm which is theoretically optimal, maximizing the recovery probability given a fixed expert allocation. In this way, our allocation and placement plan strikes a balance between the two goals.

\noindent
\textbf{Expert allocation.} For ease of presentation, we show how expert replicas are allocated and placed on each node. If a node has multiple GPUs, \sys simply distributes the assigned replicas among all GPUs on that node, as we consider failures at the node level.
We denote the number of nodes as $N$, the number of experts as $E$, the number of replicas each node can hold as $c$, the total number of tokens routed to expert $e$ as $t_e$, the number of replicas assigned for expert $e$ as $r_e$. To speed-up computation, we want the ratio of replicas assigned to each expert match the ratio of tokens routed to that expert, namely $\frac{r_e}{\sum_{e^{\prime}}r_{e^{\prime}}} \approx \frac{t_e}{\sum_{e^{\prime}} t_{e^{\prime}}}$. 
Furthermore, for better fault tolerance, we define a fault-tolerant threshold $f$, where \sys guarantees a 100\% probability of failure recovery as long as fewer than $f$ nodes fail simultaneously. Hence, each expert is assigned at least $f$ replicas.

Assume that the experts are sorted by the number of routed tokens ($t_e$) in ascending order. We iteratively compute the number of replicas $r_e$ assigned to each expert $e$ as follows:
\begin{align}
    r_e=\mathsf{max}\{\lfloor \frac{t_{e}}{\sum_{e^\prime=e}^E t_{e^\prime} }  \cdot( N \cdot c -\sum_{e^\prime=1}^{e-1} r_{e^{\prime}}) \rfloor, f \}
\end{align}
Our assignment strategy ensures that $\sum_e r_e=N\cdot c, r_e\ge f, r_e\ge r_{e-1}$. $\frac{r_e}{\sum_{e^{\prime}} r_{e^{\prime}}} \approx \frac{t_e}{\sum_{e^{\prime}} t_{e^{\prime}}}$ is also satisfied in most cases for training speed-up under the imbalanced workload.
As $r_e \ge f$, \sys guarantees recovery under failures of a small number~$(<f)$ of nodes. 

\noindent
\textbf{Expert placement.} However, when the number of failed nodes $\ge f$,  the probability of failure recovery differs between different placement plans.
Figure~\ref{fig:recover_exp1} shows an example of 4 experts and 5 nodes. Assume that 3 nodes fail  simultaneously. In plan A, the probability of recovery is $\frac{5}{10}$, as recovery is possible only if the alive nodes are $(1,2),(1,3),(1,4),(2,4),(2,5)$, while there are 10 possible cases. In plan B, however, the probability of recovery is much higher at $\frac{7}{10}$.

\begin{figure*}
\centering
\includegraphics[width=.9\linewidth]{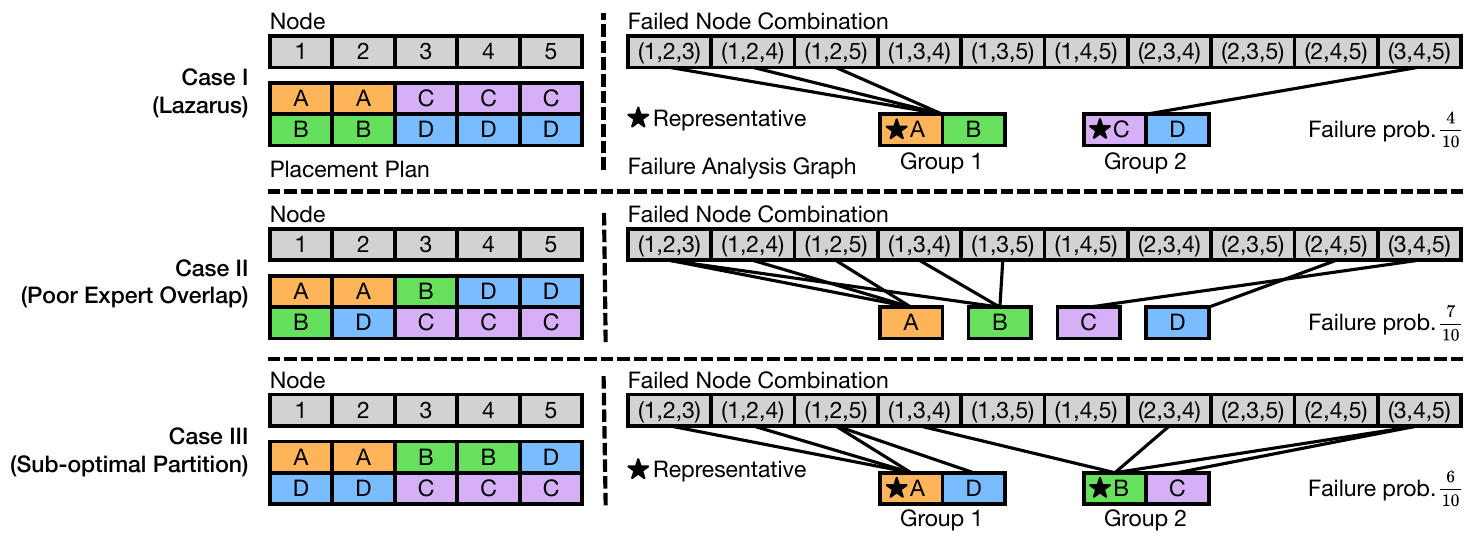}
\caption{\sys minimizes the failure probability by minimizing the number of vertices representing node failures that have incident edges. Here we consider 3 node failures. Comparing Case I and II, when expert overlap on nodes is not maximized, there are more unique failure patterns. Comparing Case I and III, swapping any expert also leads to more failure patterns.}
\label{fig:recover_exp2}
\VSPACE{-4mm}
\end{figure*}

\noindent
\textbf{Placement solution for an easier case.} We first consider a simpler case where $E\le c$, which we can easily derive an optimal placement strategy inspired by the previous example.
The strategy is that for the first $\mathsf{min}(r_1,N)$ nodes we place the first (least popular) expert, for the first $\mathsf{min}(r_2,N)$ nodes we place the second expert, and so on. For the vacant positions, we 
uniformly place the experts that still have replicas left. In this way, denote the set of nodes that have the $e$-th expert as $S_e$. This strategy satisfies $S_1 \subset S_2\cdots \subset S_E$. Thus, the recovery probability is equal to the probability that the first expert belongs to an alive node (i.e., any of the first $r_1$ nodes are alive). Furthermore, the first expert belonging to an alive node is a necessary condition of failure recovery. Thus, for any placement plan, the recovery probability is upper bounded by the probability of any of the first $r_1$ nodes is alive. Since there are only $r_1$ replicas for the first expert, it can span across at most $r_1$ different nodes. 
Therefore, in the case of $E \le c$, this placement strategy achieves the upper bound of the recovery probability of all placement plans, guaranteeing its optimality. 

The above strategy relies on a core principle: maximize the nodes overlapped between the experts.
Take the first and second expert as an example; if we overlap all the replicas of the first experts with the second expert's replicas on the same nodes, the two experts' states will be lost only when all of the first expert's replicas are lost. However, if some of the first expert's replicas are not overlapped with the second expert's, the two experts cannot be recovered when either all of the first's replicas are lost or all of the second's are lost. 

\noindent
\textbf{Placement solution for the more difficult case.} 
When $E>c$, the optimal strategy becomes more complicated. The previously introduced maximum overlap principle cannot be directly applied due to 
the infeasibility of overlapping all experts when $E>c$. To address this issue, we partition both the experts and nodes into $\lceil \frac{E}{c} \rceil$ groups. We also modify the second principle into maximizing the overlap of experts in each group. Furthermore, we constrain the expert partitions to be consecutive, i.e., $1,\cdots,c$ are in the first group, $c+1,\cdots, 2c$ forms the second group, and so on. 
For the nodes, we divide the first $\min\{N,\sum_{i=1}^{\lceil \frac{E}{c} \rceil} r_{c*(i-1)+1} \}$ nodes into $\lceil\frac{E}{c} \rceil$ groups. The first group has $r_1$ nodes, the second has $r_{c+1}$ nodes, and so on, while the last group has $\min\{N-\sum_{i=1}^{\lceil \frac{E}{c} \rceil-1} r_{c*(i-1)+1}, r_{c*(\lceil \frac{E}{c} \rceil-1)+1}\}$ nodes. For group $i$ in the first $\lceil \frac{E}{c} \rceil-1$ groups, each node contains one replica of expert $c*(i-1)+1,\ldots, c*i$. For the last group, each node contains one replica of expert $c*(\lceil \frac{E}{c} \rceil-1)+1,\ldots, E$. For the vacant slots, we uniformly place the experts that still have replicas left to place.  Our strategy satisfies $S_{c*(i-1)+1} \subset S_{c*(i-1)+2}\cdots \subset S_{c*i}$ for different $i$, which intuitively maximizes the node overlap of experts in each group. The recovery of the experts in the $i-$th group hence only requires one node in $S_{c*(i-1)+1}$ to be alive, where we define the expert $c*(i-1)+1$ as the representative of group $i$. The complete recovery is equivalent to that one replica of each group's representative still remains. Our maximum rank overlap (MRO) placement plan maximizes recovery probability under uniformly random node failures for any given replica number $r$. Concretely, we have Theorem~\ref{theorem:mro_optimality}. Its proof can be found in the supplementary material.

\begin{theorembody}
For any MRO plan $T$ and $R$, given the number of replicas $r_e$ for each expert $e$, $T$ maximizes the recovery probability $\mathsf{Pr}(\bigcup_{a \in A}Col_a=[E])$, where $[E]$ is the set of experts, $Col_a$ is the set of replicas assigned to node $a$, $A$ is a uniformly sampled set of $R$ nodes that remain alive. 
\label{theorem:mro_optimality}
\end{theorembody}

The two core insights of our method are to partition experts into different groups based on their popularity, and maximize the overlap across experts in the same group. 
Here, we offer some intuition for analyzing the optimality of our method. 

Denote the failed node number as $k$; there are $\binom n k$ different combinations of node failure cases. The failure analysis of different placement plans would be much more explicit by visualizing a graph formed this way: It is a bipartite graph, where one set has $E$ vertices, each representing an expert. The other set has $\binom{n}{k}$ vertices, each corresponding to a case of node failure. Each placement strategy can be analyzed by constructing such a bipartite graph: for every expert in the placement, an edge is created between that expert and every node failure combination that renders the expert unrecoverable.

Essentially, our strategy achieves optimality by ``putting all the eggs in one basket.'' 
The failure probability of a given placement strategy can be counted by the proportion of the second set of vertices (failure combinations) that have incident edges.
Our strategy minimizes the number of the second sets of vertices that have incident edges. This is achieved by forcing experts to share the same failure combinations as much as possible,  through maximizing the overlap between popular experts.
For each non-representative expert, the set of failure combination vertices it connects to is a subset of the vertices the representative expert connects to.
For different placement plans, the number of edges is the same. By letting more experts fail under a smaller set of failure cases, we reduce the number of failure nodes that have incident edges, thus improving recovery probability.

When the overlap between experts is not maximized, as in case II of \autoref{fig:recover_exp2}, expert $B$ creates unique failure combinations $(1,3,4),(1,3,5)$, expert $D$ creates a unique failure combination $(2,4,5)$. Therefore, the failure probability rises to $\frac{7}{10}$, compared to the optimal of $\frac{4}{10}$.

In addition, case III in \autoref{fig:recover_exp2} offers an example of the optimality of our group partition strategy. If any swap is conducted in two expert groups, where the second expert group has a more popular representative, the total number of failed combinations connected by the two groups' representatives increases, because the new representatives are less popular than the old representatives.

We note that the expert load distribution can be different across layers, hence we compute an expert replica allocation and placement plan independently for each layer. As the load distribution also shifts during training according to the workload, \sys also periodically rebalances the expert allocation and updates the placement plan.

Now, we have developed the strategy to assign expert replicas to each node (GPUs). Next, we explore under such asymmetric replica placements, how \sys efficiently dispatches tokens to GPUs with replicas of routed experts.

\subsection{Flexible Token Dispatcher}
\label{sec:token-dispatch}

\begin{algorithm}[t]
\SetAlgoLined
\SetAlgoNoEnd
\DontPrintSemicolon
\SetKwInOut{Input}{Input}
\SetKwInOut{Output}{Output}
\newcommand\myCommentStyle[1]{\small\textcolor{gray}{#1}}
\SetCommentSty{myCommentStyle}
\SetKwComment{Comment}{// }{}
\newcommand{\forexpert}{$e\leftarrow 0$ \KwTo $E$ in parallel}
\newcommand{\forrank}{$j\leftarrow 0$ \KwTo $N$ in parallel}

\Input{$N$: Number of GPUs; $i$: Current GPU rank; 
$h$: Activation of input tokens to the MoE block;
$R_{e,j}$: Number of replicas for expert $e$ assigned to rank $j$; $T_{e,j}$: Number of tokens routed to expert $e$ at rank $j$;
}
\Output{$h^{\prime}$: Shuffled inputs for all-to-all dispatch; $s_j$: Number of tokens to dispatch to rank $j$ }
\For{\forexpert}{
    $r_e \leftarrow \sum_j R_{e,j}$ \Comment{total \#replicas for expert $e$}
    $t_e \leftarrow \sum_j T_{e,j}$ \Comment{total \#tokens routed to expert $e$}
    $p_e \leftarrow t_e / r_e $ \Comment{\#tokens each replica should handle}  \label{eq:dispatch-tokens_per_replica}
    \For{\forrank}{
        $P_{e,j}\leftarrow c_{e} R_{e,j}$ \Comment{\#tokens rank $j$ can process} \label{eq:dispatch-rank-capacity}
        $P_{e,j}\leftarrow P_{e,j} - \min({P_{e,j}, T_{e,j}})$\; \Comment{rank $j$'s local tokens are prioritized} \label{eq:dispatch-local_priority}
    }
    $D_{e,i} \leftarrow c_{e} R_{e,i} - P_{e,i}$ \Comment{locally processed \#tokens}
    \For{$j\leftarrow 0$ \KwTo $N$, $j \neq i$ in parallel}{
        $D_{e,j} \leftarrow (T_{e,i} - D_{e,i}) \frac{P_{e,j}}{\sum_{k \neq j}P_{e,k}}$\;
        \Comment{distribute remaining tokens to other ranks} \label{eq:dispatch-tokens_to_each_rank}
    }
}
\For{\forrank}{
    $s_j \leftarrow \sum_{e^{\prime}}D_{e^{\prime},j}$ \Comment{\#tokens dispatched to rank $j$} \label{eq:dispatch-send_size_to_each_rank}
    \For{\forexpert}{
        \label{eq:dispatch-reshuffle_start}
        $start \leftarrow  \sum_{0..j-1}s_{j^\prime} + \sum_{0..e-1} D_{e^\prime, j}$\;
        $end \leftarrow \sum_{0..j-1}s_{j^\prime} + \sum_{0..e} D_{e^\prime, j}$\;
        $h^\prime [start..end] \leftarrow$ ($\sum_{j^\prime=0}^{j-1} D_{e,j^\prime}$)-th to ($\sum_{j^\prime=0}^{j} D_{e,j^\prime}$)-th tokens in $h$ that routed to $e$\;
        \label{eq:dispatch-reshuffle-end}
    }
}
\Return{$h^{\prime}, s$}
\caption{\label{alg:token-dispath} Token dispatch algorithm.}
\end{algorithm}

In traditional expert parallelism, each token can be simply dispatched to the GPU that owns the corresponding expert, as there is only a single replica for each expert within a particular EP group. Concretely, an all-to-all is performed with all ranks (GPUs) in the EP group sending and receiving the same number of tokens, which is dynamically set to the maximum number routed to a single expert to prevent token dropping, while unused slots are padded~\cite{rajbhandari2022deepspeed,hwang2023tutel}. 

With \sys's adaptive expert placement, there are varying numbers of replicas assigned for each expert on different sets of GPUs. Therefore, each rank must decide which rank with the routed expert's replica to dispatch a token to. The fact that multiple replicas can be assigned to the same rank (indicating more tokens should be dispatched to it), combined with the difference in expert routing on different ranks, a challenge emerges --- how can we efficiently dispatch the tokens to all GPUs with the routed experts while balancing the load? If tokens are poorly dispatched, some ranks could receive significantly more tokens than others, hence defeating the purpose of our adaptive expert allocation. Moreover, the padded all-to-all is no longer viable in our case where a token can be dispatched to any rank (instead of within an EP group), as padding would dominate the communication.

To address these issues, we design a flexible token dispatcher that efficiently dispatches each token to a particular GPU and balances the number of tokens routed to each GPU. With the dispatch schedule computed, \sys performs a flexible all-to-all without padding. Algorithm~\ref{alg:token-dispath} shows the workflow of the token dispatcher, which is implemented in a CUDA kernel to process all experts and target ranks in parallel. The basic idea behind Algorithm~\ref{alg:token-dispath} is that each replica of an expert should compute around the same number of tokens, and each rank should utilize its local processing ``capacity'' before dispatching remaining tokens to other ranks.

Before computing the dispatch schedule, an all-gather is first performed to collect how many tokens are routed to each expert from all ranks, i.e., $T_{e,j}$. $T_{e,j}$ is collected so that the token dispatcher can better balance the load to each GPU based on the expert routing distribution of all tokens from all ranks, instead of using only locally computed tokens. In addition, since collective communication operations require synchronization of all participant ranks, $T_{e,j}$ is also necessary in computing how many tokens a rank should receive from each of the other ranks. Since only $E$ integers are collected from each rank, this extra all-gather imposes negligible overhead, as demonstrated in \autoref{sec:exp_single_layer_ablation}. The number of replicas allocated to each GPU $R_{e,j}$ from the placement plan is also passed to the token dispatcher.

After $T_{e,j}$ is collected, each rank $i$ independently computes how many tokens it dispatches to each of all $N$ ranks, for all $E$ experts. First, for each expert $e$, the number of tokens each replica should process is computed in line~\ref{eq:dispatch-tokens_per_replica} by evenly distributing all $t_e$ tokens routed to $e$ onto all $r_e$ replicas. The processing capacity of each rank $j$ can then be computed by multiplying $p_e$ with the number of replicas of $e$ that $j$ is assigned (line~\ref{eq:dispatch-rank-capacity}). This capacity will be prioritized towards tokens computed locally on $j$. After the remaining capacities of all ranks are computed, rank $i$ dispatches the remaining ($T_{e,i} -D_{e,i}$) tokens that are beyond $i$'s local processing capacity. The number of tokens $D_{e,j}$ to dispatch to each rank for $e$ is calculated based on their residual capacities (line~\ref{eq:dispatch-tokens_to_each_rank}). 

Since the all-to-all collective operates on a continuous buffer, the token dispatcher has to reshuffle the input activations $h$ to the MoE block, so that tokens routed to the same expert and dispatched to the same rank are grouped together. The total tokens $s_j$ to dispatch to rank $j$ across all experts is computed in line~\ref{eq:dispatch-send_size_to_each_rank}. In lines~\ref{eq:dispatch-reshuffle_start}-\ref{eq:dispatch-reshuffle-end}, these tokens are sorted by their routed experts and placed consecutively in $h^{\prime}$. The reshuffled activations $h^{\prime}$ are then used in the dispatch all-to-all collective, with $s_j$ tokens sent to each rank $j$. The token dispatcher also computes how many tokens to receive from each rank $j$ in the all-to-all in a similar fashion.

At this point, \sys has the ability to adaptively assign expert replicas and dynamically dispatches tokens among replicas of routed experts. Next, we discuss how \sys efficiently migrates to a new configuration upon failures.

\subsection{Efficient Reconfiguration}
\label{sec:effcient-reconfig}
As discussed in \autoref{sec:expert-placement}, if at least a single replica of each expert still remains after failures, \sys can recover without restarting from checkpoints. However, the remaining expert replicas' distribution could deviate from the desired allocation computed for the workload, and their placement may be prone to subsequent failures.
Therefore, \sys must reallocate expert replicas and efficiently migrate to a new placement. 
Such migration is also required when \sys rebalances the expert allocation and when new nodes join.

The ordering of nodes in the placement plan is not enforced in the placement algorithm, as long as each node in the plan maps to a physical node in the cluster. However, when migrating from an old placement plan, such a mapping becomes relevant. It directly determines how many experts' states a node needs to retrieve from other nodes, as only newly assigned ones not in the old placement plan have to be fetched. To reduce the number of replicas to shuffle during migration, hence the communication, \sys applies a greedy algorithm that iteratively maps a physical node to a node in the new placement plan that the number of newly assigned experts is minimized.

After the node mapping is determined, \sys schedules the transfers of expert states. Each node fetches missing states for the newly assigned experts from other nodes that own them. If multiple nodes require the states of the same expert, \sys distributes their state transfers among all owning nodes, to minimize the overall migration time.
\VSPACE{-2mm}

\section{Implementation}
\sys is implemented in 4K LoC in Python and 500 LoC in CUDA, building on top of PyTorch~\cite{imambi2021pytorch}~(v2.3) and using components from DeepSpeed~\cite{rasley2020deepspeed}~(v0.13). 

\noindent
\textbf{\sys controller and agents.} We implement the controller and agents using Python's asynchronous framework. New agents register with the controller, using a TCP socket for communication. The controller maintains a global view of node availability, where agents periodically send heartbeats for it to detect failures. Upon failures or scaling up with newly arrived nodes, the controller computes an updated expert placement plan, which is sent to each agent and relayed to the worker process that uses \sys runtime. The agents also periodically collect expert routing history from each worker and send it to the controller for expert rebalancing.

\noindent
\textbf{\sys runtime.} Based on the controller's configuration, our runtime sets up NCCL~\cite{nccl} communication groups for expert and non-expert gradients all-reduce, as well as all-to-all in expert computation. We implement data parallelism and expert parallelism with our adaptive expert placement;
however, \sys can be extended to combine with pipeline parallelism using techniques like Oobleck~\cite{jang2023oobleck}, which are orthogonal and complementary to ours. Upon failures, enqueued NCCL operations time out and the model states are not updated on the failed step, while a new configuration is received from the agent via a listener thread. Batched NCCL send/recv primitives are used to transfer states during migration. For scaling up and rebalancing, \sys performs reconfiguration lazily, only after the current training step is finished.
\VSPACE{-2mm}
\section{Evaluation}
\subsection{Setups}
\label{sec:exp-setup}
\noindent
\textbf{Testbed.} We have five servers in our testbed, each with 2 NVIDIA RTX 3090 GPUs and a 100\,Gbps Mellanox ConnectX-5 NIC connected to a single 100\,Gbps Mellanox SN2100 switch. Due to limited resources, we treat each GPU as a separate node to emulate a cluster of 10 GPUs. To store checkpoints, we deploy a NFS server on a separate machine, which is connected to the GPU servers via 10\,Gbps NICs.

\noindent
\textbf{Baselines.}
As there is no existing system to support resilient and elastic training of MoE models, we compare \sys against a checkpoint-based baseline using DeepSpeed MoE (DS)~\cite{rajbhandari2022deepspeed}, a widely adopted MoE training system with both system-side and model design-side optimization. To evaluate \sys's adaptive expert placement algorithm and flexible token dispatcher, we also build a fault tolerant baseline based on DeepSpeed MoE, utilizing efficient reconfiguration module from \sys runtime. We denote this baseline as DS(FT). Similar to \sys, if a complete replica of all experts still exists upon failures, it reconfigures the workers (reassigns EP groups) and retrieves required model (expert) states from owning nodes. 

\begin{table}
\centering
\caption{Configurations of models used in the evaluation.}
\begin{tabular}{lccc}
    \toprule
     &  GPT-S & GPT-M & GPT-L \\
    \midrule
    \# Layers & 12 & 12 & 12 \\
    Feature dim. &  768 & 1024 & 1024\\ 
    \# Experts & 8 & 12 & 16 \\
    \# Params & 521M & 1.3B & 1.7B \\
    \bottomrule
\end{tabular}
\label{tab:workload-setup}
\VSPACE{-4mm}
\end{table}

\noindent
\textbf{Workloads.} Based on the widely used GPT-2 architecture, we adopt three MoE models of varying sizes and number of experts, listed in \autoref{tab:workload-setup}. We use a per-GPU batch size of 4 and a sequence length of 1024 following GPT-2's setup~\cite{fedus2022switch}. For all evaluations, we use Wikitext-2 dataset~\cite{merity2016pointer}, top-1 gate, and FP16 precision for training.

For reproducibility, we use the routing history trace from SmartMoE~\cite{zhai2023smartmoe} artifact to emulate gate networks' routing decisions. We use the loads of the top experts at each layer to construct a routing trace for each of the models we evaluate.
We set the number of expert replica slots for each GPU to 6, which is the upper limit based on available GPU memory.
With DS's traditional expert parallelism, GPT-M can fully utilize all slots, while GPT-S and GPT-L can only use 4, as the multiple of slots per GPU and EP size must equal to the number of experts. GPT-S and GPT-M can utilize an EP size of 2, hence DS and DS(FT) fully utilize all 10 nodes in the cluster, while with 16 experts and an EP size of 4, they can only utilize 8 nodes on GPT-L. We set the checkpoint interval to every 50 steps for DS and every 250 steps for DS(FT), unless mentioned otherwise. We set the minimum replicas per expert ($f$) to 2 for \sys so that recovery is guaranteed under common single node failure scenarios. \sys rebalances expert replica allocation every 200 steps.

\VSPACE{-2mm}
\subsection{Controlled Single Node Failures}
\label{sec:control-single}

\begin{figure*}
\centering
\includegraphics[width=0.95\linewidth]{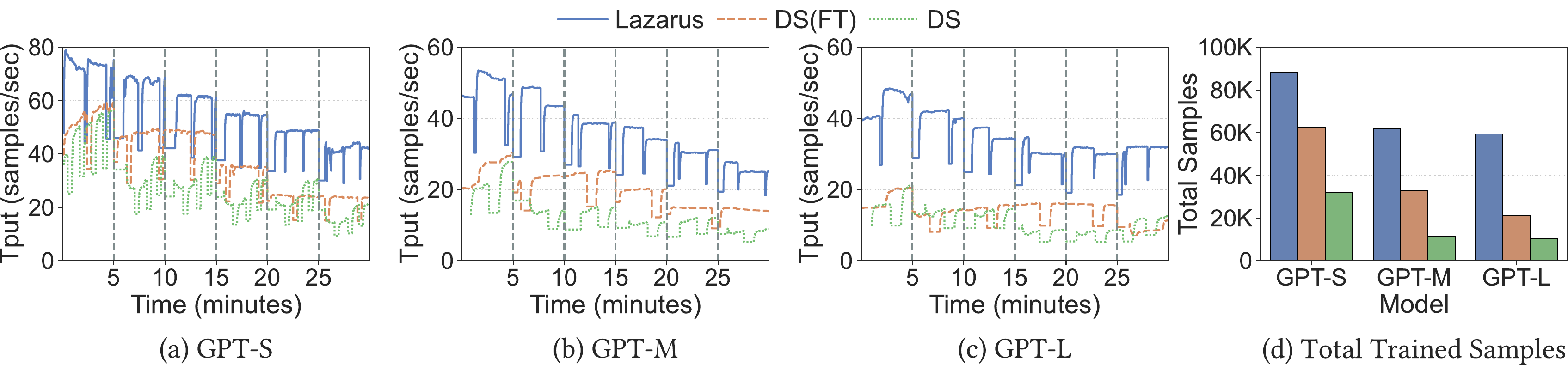}
\caption{\textbf{[Single node failure]:} Throughput and total trained samples with a single node fails every 5 minutes. DS refers to checkpointing-based DeepSpeed MoE; DS(FT) is a fault tolerant version we build using components from \sys's runtime.}
\label{fig:control-5min}
\VSPACE{-2mm}
\end{figure*}

\begin{figure*}
\centering
\includegraphics[width=0.95\linewidth]{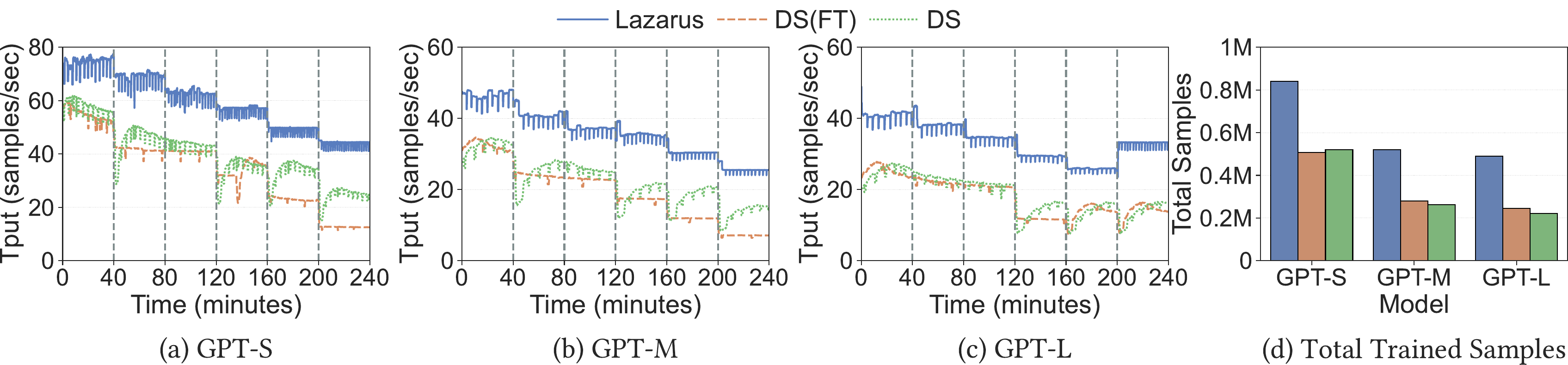}
\caption{\textbf{[Single node failure]:} Throughput and total trained samples with a single node fails every 40 minutes.}
\label{fig:control-40min}
\VSPACE{-2mm}
\end{figure*}

We first evaluate the performance in a more common case where a single node fails at a time. We consider both high failure frequency and low frequency scenarios, where we randomly choose a node to fail every 5 or 40 minutes, until only half of the nodes remain. The same set of nodes is selected to fail in each run for fair comparison. The results are shown in \autoref{fig:control-5min} and \autoref{fig:control-40min}. The throughput is smoothed over a short time window for visibility. The fluctuation in \sys's throughput is caused by the reconfiguration after node failures and the periodic rebalance of expert allocations, while the fluctuation in DS and DS(FT) is caused by checkpointing, restarting and reconfiguration (only for DS(FT)). To reduce the overhead of checkpointing for DS and DS(FT) in the low failure frequency (40 minutes) setting, we increase their checkpoint intervals by 4x to 200 steps and 1000 steps, respectively. We also note that using such low checkpoint frequency would prevent DS from making any effective progress under high failure frequency (5 minutes). 

From \autoref{fig:control-5min} with a failure frequency of 5 minutes, we observe that over the 30 minutes duration of training, \sys finished a total of 2926 and 1996 steps, trained 2.8x and 5.7x samples on GPT-S and GPT-L, compared with DS. The performance gains significantly increase on GPT-L, as the checkpointing and restarting overhead grows with model sizes. Moreover, as in the GPT-L setting (EP size is 4), 4 nodes are required to hold a complete replica of all experts for DS and DS(FT), they can only utilize either 4 or 8 nodes, while they can utilize all 10 nodes at the start for GPT-S and GPT-M.

\sys also outperforms DS(FT) by 1.4x and 2.8x on GPT-S and GPT-L. On the smaller GPT-S, there are a large number of replicas for each expert (5 replicas initially), hence DS(FT) can recover in each failure. However, as the number of experts and EP size increases on GPT-L, DS(FT) has to restart from checkpoints after failures of both EP groups.

When the failure is infrequent as shown in \autoref{fig:control-40min}, the performance difference between \sys and DS decreases as the overhead of checkpointing and restarting decreases. Still, \sys outperforms DS by 1.6x and 2.3x on GPT-S and GPT-L. As the overhead of DS decreases, DS and DS(FT) have similar performance in this case.

We also observe that \sys's throughput tends to monotonically decrease as the number of nodes decreases, as \sys can fully utilize all remaining nodes for training. While for DS and DS(FT), the throughput experiences steep drops since they can only utilize a multiple of EP size nodes. We note that the throughput of \sys increases in the last 40 minutes in \autoref{fig:control-40min}. This is because \sys no longer enforces a minimum of 2 replicas for each expert for fault tolerance, as there are not enough slots with 5 nodes left.

\sys still outperforms DS by a great margin, even when both of them fully utilize all 10. For instance, for GPT-M, during the first 40 minutes in \autoref{fig:control-40min} when no node fails, \sys has a throughput of 45 samples/sec during effective computation (factoring out checkpoint and rebalance overheads), while DS only reaches 34 samples/sec.

From \autoref{fig:control-5min}, we also observe that compared to DS, DS(FT) sometimes has higher throughput during effective computation. For instance, for GPT-M, DS(FT) outperforms DS by 1.6x during the 5\textasciitilde10 minutes window, when they all fully utilize the remaining 8 nodes. This is mainly caused by the highly imbalanced expert loads during the early periods of training, while DS(FT) progresses much faster without the overhead of checkpoint and restarting. When we increase the checkpoint intervals by 4x for both baselines in \autoref{fig:control-40min}, together with the lower failure frequency, such divergence disappears. Instead, for GPT-S and GPT-M, DS(FT) is slower than DS during the last 80 minutes. In these two cases, DS(FT) always resumes training by reconfiguring currently used nodes that are still alive. It does not use previously dropped nodes (due to exceeding EP size of 2), while DS attempts to utilize all nodes it can when restarting.

Overall, checkpointing and restarting overhead becomes increasingly significant with larger models and higher failure frequency. Comparing with DS(FT) which shares \sys's efficient reconfiguration runtime, \sys's adaptive expert placement improves both training throughput and resiliency.

\VSPACE{-2mm}
\subsection{Controlled Multi Node Failures}

\begin{table}
\resizebox{0.99\columnwidth}{!}{
    \begin{tabular}{lcccc}
        \toprule
         & \multicolumn{2}{c}{GPT-S} & \multicolumn{2}{c}{GPT-L} \\
        \cmidrule(lr){2-3} 
        \cmidrule(lr){4-5}
         & step 200 & step 4000 & step 200 & step 4000 \\
         \midrule
         \# Lost nodes & 2 & 3 & 4 & 5 \\
         \midrule
         Reconfig time (s) & 21.3 & 34.1 & 18.2 & 19.7 \\
         \midrule
         \# Experts transfer & 11 & 52 & 160 & 55 \\
         \midrule
         Transfer time (s) & 2.3 & 3.0 & 7.6 & 7.8 \\
        \bottomrule
    \end{tabular}
}
    \caption{\textbf{[Multi-node failures]:} Recovery overhead of \sys under multiple node failures on sampled cases.}
    \label{tab:recovery-overhead}
    \VSPACE{-4mm}
\end{table}

\begin{figure}
    \begin{subfigure}{0.48\linewidth}
        \centering
        \includegraphics[width=\linewidth]{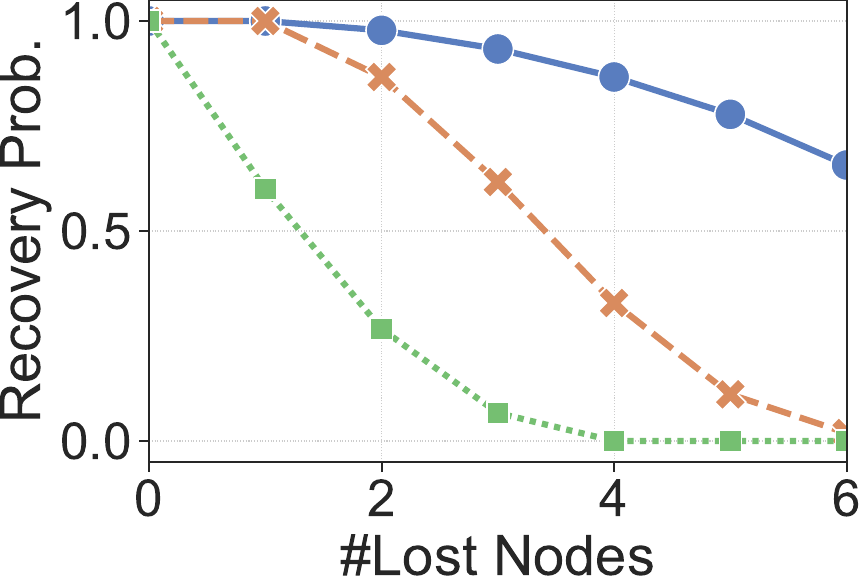}
        \caption{GPT-S (step 4000)}
        \label{fig:recovery_prob_multi_fail_gpt_s}
    \end{subfigure} \hfil
    \begin{subfigure}{0.48\linewidth}
    \centering
      \includegraphics[width=\linewidth]{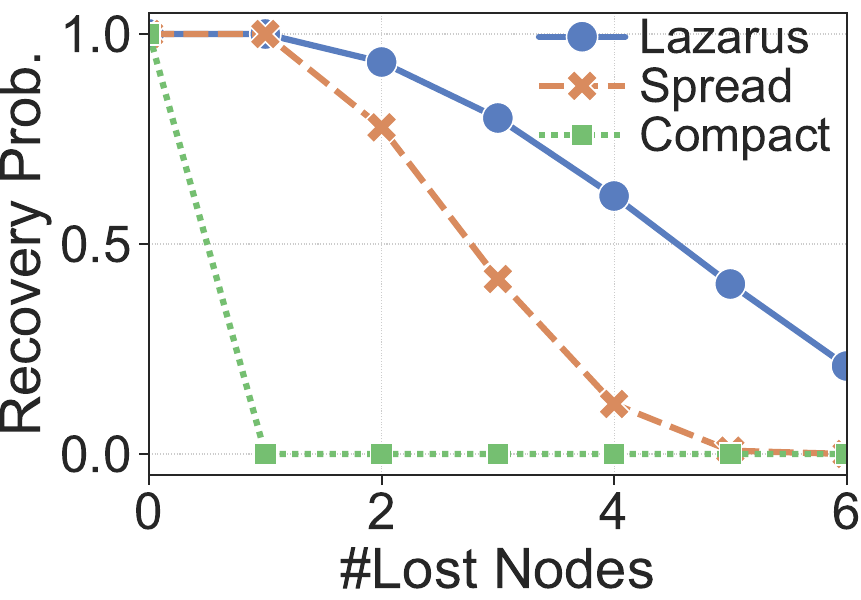}
        \caption{GPT-L (step 200)}
        \label{fig:recovery_prob_multi_fail_gpt_l}
    \end{subfigure}
\caption{\textbf{[Multi-node failures]:} Recovery probabilities using different expert placement strategies.}
\label{fig:recovery_prob}
\VSPACE{-4mm}
\end{figure}

\begin{figure*}
\centering
\includegraphics[width=0.94\linewidth]{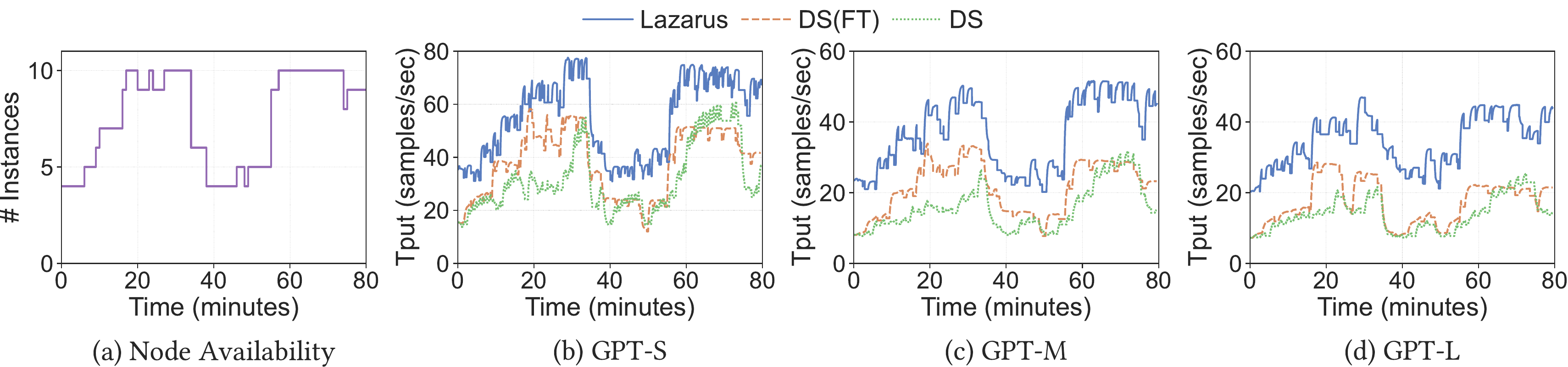}
\caption{\textbf{[Spot instance]:} Throughput changes in spot instance environment.}
\label{fig:spot-trace}
\end{figure*}

Next, we study how well \sys handles simultaneous failures of multiple nodes. Whether \sys can recover from such failures depends on both the expert allocation (i.e., how many replicas are assigned to each expert) and expert placement, as well as which concrete set of nodes fail.  The allocation and placement changes as the expert load distribution varies over the duration of training, and it is also different for different layers. Hence, we evaluate \sys's system overhead of recovery by sampling several cases for GPT-S and GPT-L at different training steps, while we evaluate \sys's placement algorithm by computing the recovery probability for a model at a given training step. The recovery probability can be computed by enumerating all possible combinations of failed nodes, since the way experts are allocated and placed only depends on the expert load at the particular step. 

The recovery overhead for sampled cases is shown in \autoref{tab:recovery-overhead}, where 2 to 5 nodes are selected to fail at training step 200 and 4000. We report the total number of experts replicas that need to be transferred between nodes and the time spent on the state transfers. The weights and optimizer states of each is 63MB for GPT-S and 112MB for GPT-L. We find that the overhead of state transfers is negligible. This low overhead is mainly contributed by the fact that required states can be fetched from other nodes instead of the much slower remote storage, and \sys balances the point to point send/recv operations among all owning ranks of an expert's states. We also report the total reconfiguration time, from failure occurrence to training resumption, where state transfers only constitute a small portion. Throughout our entire evaluation, we find that each reconfiguration event takes 20\textasciitilde40 seconds. It takes 10\textasciitilde20 seconds for enqueued NCCL kernels to time out and 5\textasciitilde15 seconds for reconfiguring NCCL's communication groups. We also observe that the placement plan's computation takes less than 100ms.

To demonstrate the effectiveness of \sys's fault-tolerant expert placement algorithms, we compare it with two baselines: a spread placement strategy which distributes each expert's replicas across different nodes in a round-robin fashion, and a compact strategy that packs an expert's replicas on a minimum number of nodes. The recovery probabilities with respect to the number of nodes failed are illustrated in \autoref{fig:recovery_prob}. We find that \sys's placement algorithm greatly outperforms both baselines. For instance, for GPT-L at step 200, \sys has a 41\% recovery probability with 4 node failures, compared to 12\% of spread placement. We also observe that on the smaller GPT-S when expert loads are relatively more balanced at later step 4000, compact placement achieves limited recovery capability with 1 or 2 node failures. However, it completely fails to recover in any failure scenario on the larger GPT-L with 16 experts.

\subsection{Spot Instance Trace}
\label{sec:spot-instance}
We also borrow a real spot instance node availability trace from Bamboo~\cite{thorpe2023bamboo} to evaluate \sys under both failures and scaling-up. The trace includes both preemption events and node additions. We replay a representative 80 minutes segment of the availability trace collected on AWS EC2 P3 instances. As the original trace is collected on a 32 nodes cluster, we cap the maximum number of nodes to 10 in our testbed setup. To handle rare cases where recovery is not possible due to too many nodes failing at the same time, we also apply periodic checkpointing for \sys. We set the checkpoint interval to every 250 steps, the same as DS(FT), for fair comparison. For node addition events, all compared methods waited for 2 minutes to accumulate sufficient nodes before scaling up, to avoid frequent reconfiguration or restarting. The results are shown in \autoref{fig:spot-trace}. 

Over the 80 minutes duration, \sys trained 2.3x and 3.4x samples on GPT-S and GPT-L, compared with DS. \sys outperforms DS(FT) by 1.2x and 1.8x on GPT-S and GPT-L. We also note that \sys's throughput changes proportionally to the number of nodes available, as \sys wastes no node, while DS and DS(FT) are limited by EP sizes. 

Due to the overhead of checkpointing and restarting, DS trained 51\% and 48\% fewer samples than DS(FT). DS(FT) can always recover from failures for GPT-S and GPT-M, as it evenly allocates up to 5 replicas to all experts at a cost of reduced throughput. For GPT-L, however, when there is fewer than 8 nodes, DS(FT) cannot utilize more than 4 nodes for redundancy. It has to restart from the checkpoint each time, leading to 3\textasciitilde5 minutes of lost progress.

We observe that only in a single preemption event when 4 nodes are lost at 34 minutes, \sys has to restart from checkpoint. Note that in the original trace, only a maximum of 19\% nodes failed at a time.

\VSPACE{-2mm}
\subsection{Ablation Study}
\subsubsection{Impacts of Expert Load Imbalance}
\label{sec:exp_single_layer_ablation}

\begin{figure}
    \begin{subfigure}{0.48\linewidth}
        \centering
        \includegraphics[width=\linewidth]{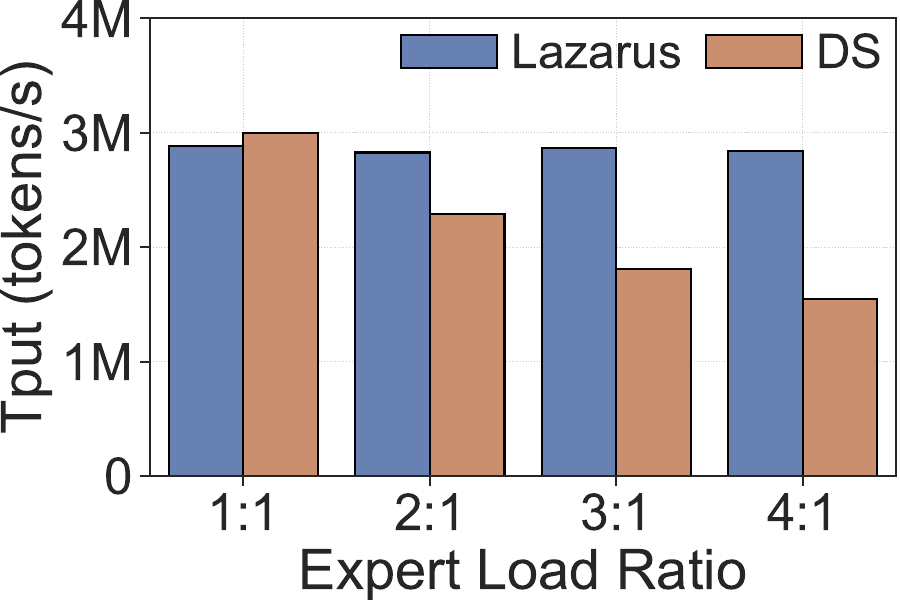}
        \caption{Throughput}
        \label{fig:single_layer_throghput}
    \end{subfigure} \hfil
    \begin{subfigure}{0.48\linewidth}
    \centering
      \includegraphics[width=\linewidth]{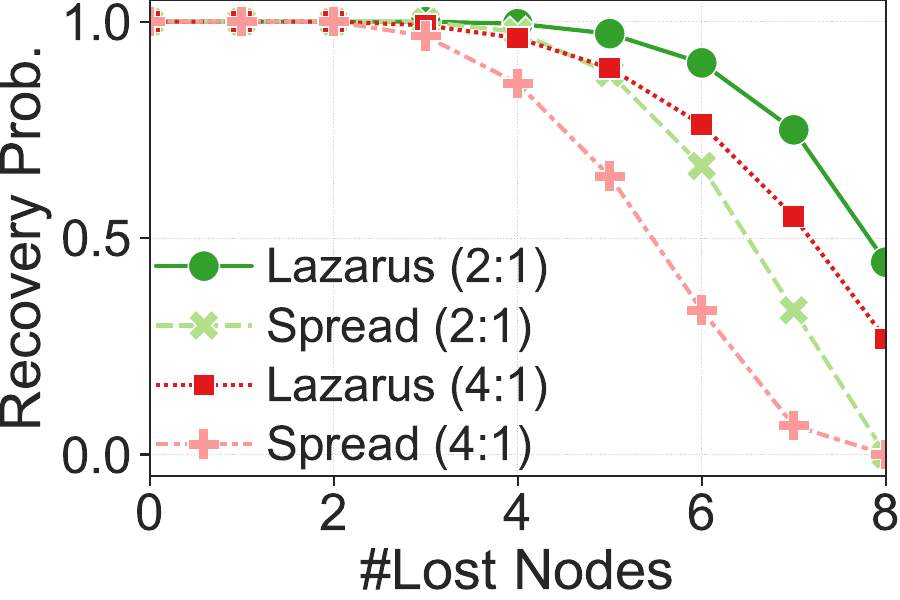}
        \caption{Recovery Probability}
        \label{fig:single_layer_recovery}
    \end{subfigure}
\caption{\textbf{[Ablation Study]:} Single layer throughput and recovery probabilities under different expert load ratios.}
\label{fig:single_layer_ablation}
\VSPACE{-2mm}
\end{figure}

To study how the expert load imbalance in workloads affects both \sys's performance and fault resiliency, we build a single MoE layer with 8 experts and a feature dimension of 1024. We construct workloads with different expert load ratios. We show the layer forward throughput in \autoref{fig:single_layer_throghput}. Here, a load ratio of 4:1 indicates that 4x more tokens are routed to one of the experts than if all experts are evenly routed to.

We observe that \sys's throughput remains constant as the load ratio changes, attributed to \sys's adaptive expert allocation based on expert load distribution. DS's throughput, however, dramatically decreases as the workload becomes more skewed. When the workload is perfectly balanced (1:1), \sys suffers a small overhead due to its token dispatcher.

We also evaluate the effectiveness of \sys's expert placement algorithm in fault tolerance as the load distribution changes. \autoref{fig:single_layer_recovery} shows the recovery probability of \sys with varying numbers of failed nodes on 2:1 and 4:1 load ratios, compared with the spread placement strategy. We observe that the recovery probability decreases with more imbalanced workload, as less popular experts are assigned fewer replicas. Still, \sys's placement algorithm is much more effective than spread placement, while our previous evaluation demonstrates the increased throughput is worth the effort of skewed expert allocation.

\VSPACE{-2mm}
\subsubsection{Running Time Breakdown}

\begin{figure}
\centering
\includegraphics[width=.95\linewidth]{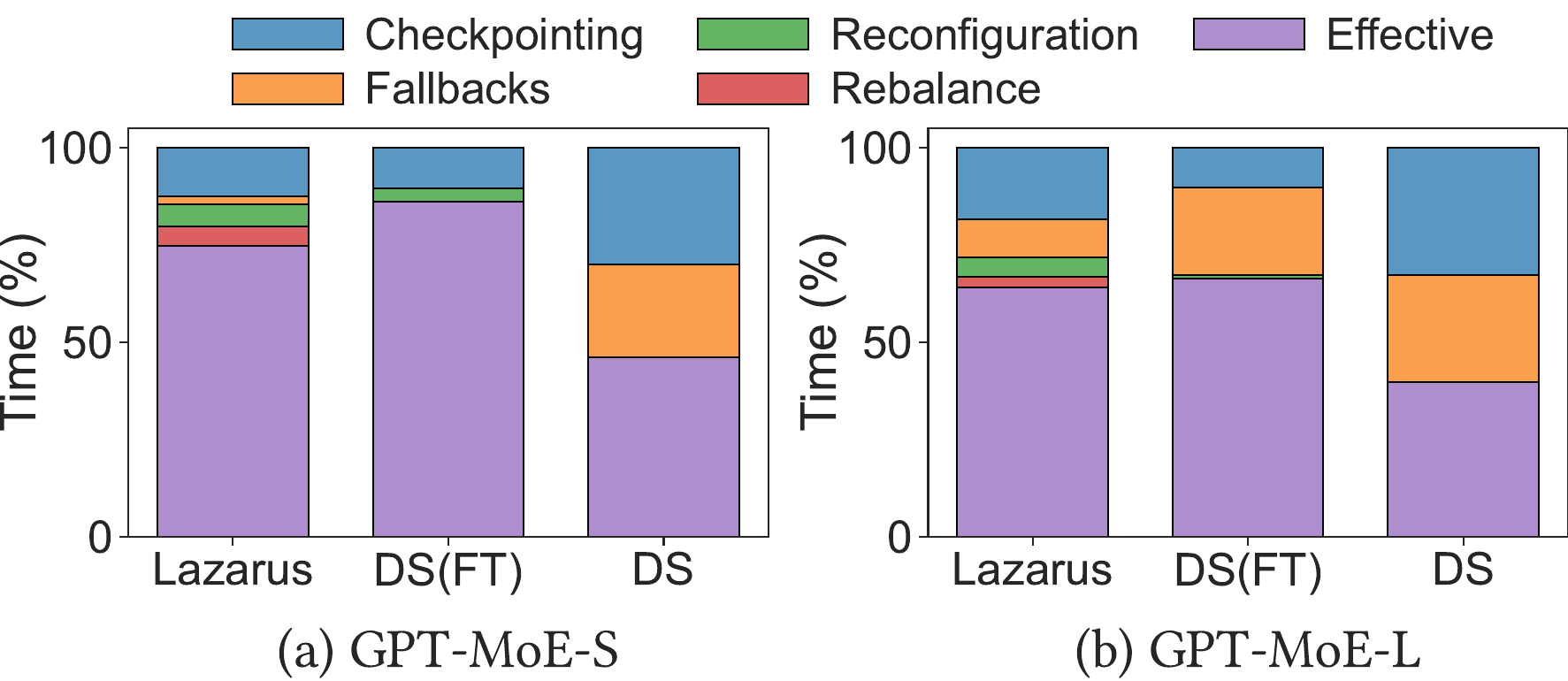}
\caption{\textbf{[Ablation Study]:} Running time breakdown of GPT-S and GPT-L on the spot instance trace.}
\label{fig:time-breakdown}
\VSPACE{-4mm}
\end{figure}

We breakdown the running time on the spot instance trace from \autoref{sec:spot-instance} in \autoref{fig:time-breakdown}. Both \sys and DS(FT) have much more time spent in effective computation, benefiting from efficient reconfiguration module in \sys runtime, while over half of the time is spent on checkpointing and restarting (fallbacks) for DS. The reconfiguration and rebalance overhead of \sys is much smaller 
than restarting, accepting for less than 10\%. We also find that DS(FT) can recover in all cases on GPT-S, yet it suffers 27\% restarting overhead on GPT-L. Despite similar effective time, \sys outperforms DS(FT) by 1.8x in terms of total trained samples, contributed by our adaptive expert allocation and flexible token dispatcher.

\subsection{Comparison with Other MoE Training Systems}
\label{sec:comparison_tutel}
\begin{figure*}
\centering
\includegraphics[width=0.99\linewidth]{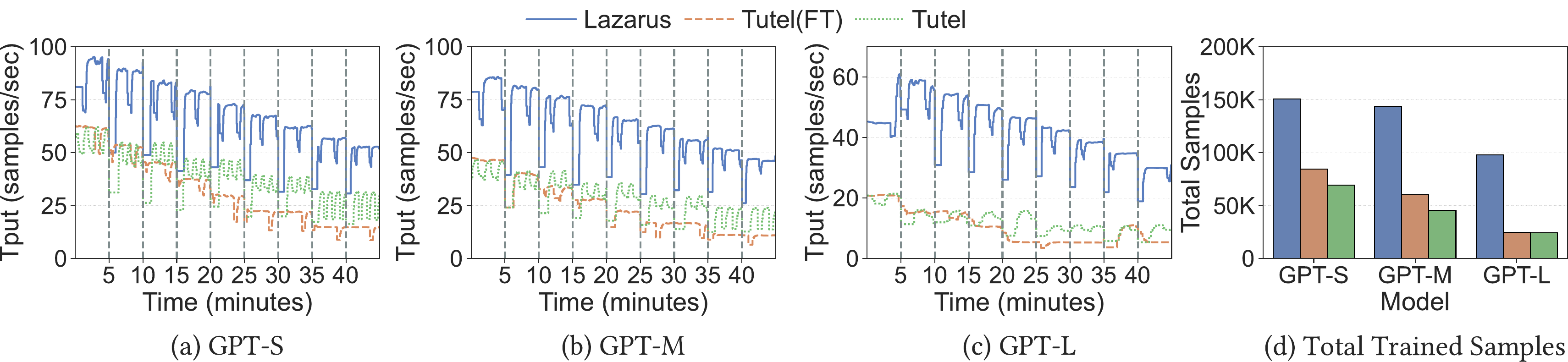}
\caption{\textbf{[Comparison with Tutel]:} Throughput and total trained samples with a single node fails every 5 minutes.}
\label{fig:aws-control-5min}
\VSPACE{-2mm}
\end{figure*}

\begin{figure}
    \begin{subfigure}{0.48\linewidth}
        \centering
        \includegraphics[width=\linewidth]{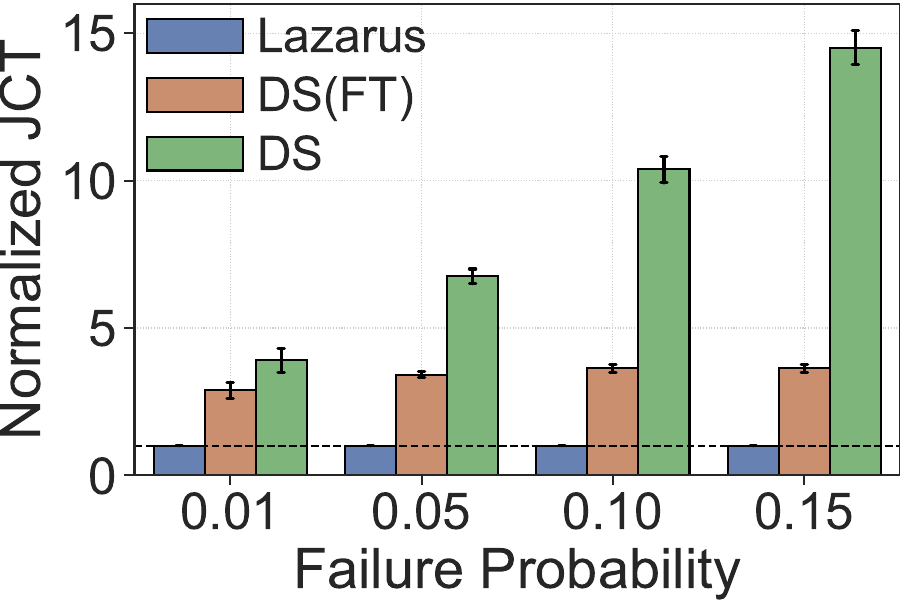}
        \caption{Different failure probability}
        \label{fig:simulation_fail_prob}
    \end{subfigure} \hfil
    \begin{subfigure}{0.48\linewidth}
    \centering
      \includegraphics[width=\linewidth]{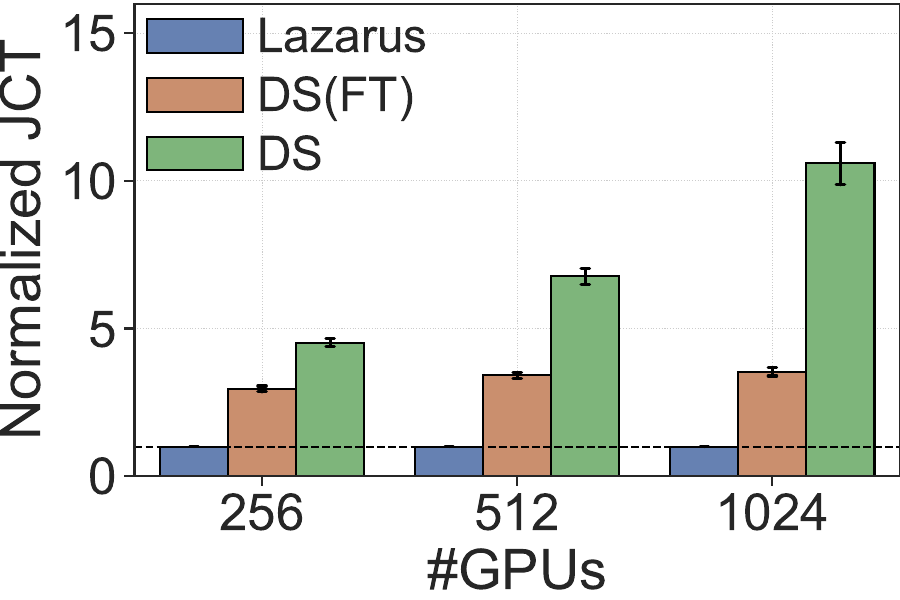}
        \caption{Different number of GPUs}
        \label{fig:simulation_num_gpus}
    \end{subfigure}
\caption{\textbf{[Simulation]:} Simulated training performance of DeepSeek~V3, under different failure probabilities with 512 GPUs and with different \#GPUs under 5\% failure probability. Error bars represent 95\% confidence intervals.}
\label{fig:simulation}
\VSPACE{-4mm}
\end{figure}

While many MoE training optimizations are orthogonal and can be directly applied to \sys to speed up training, some modify the token dispatch logic and are nontrivial to integrate. Still, \sys significantly improves end-to-end training performance with fault tolerance and elasticity. 

Here we study how \sys compares to Tutel~\cite{hwang2023tutel}, a MoE training system that implements state-of-the-art kernel optimizations for dispatch and combine operations under traditional expert parallelism. We also build a fault tolerant variant of Tutel using \sys's runtime reconfiguration module, which we denote as Tutel(FT). To demonstrate that \sys adapts to different hardware setups, we set up a testbed on AWS using 16 \texttt{g5.2xlarge} instances. Each instance has an NVIDIA A10G GPU, while the instances are connected via a 10~Gbps TCP network. To store checkpoints, we employ AWS EFS that provides shared file systems. Due to the limited network bandwidth, we also employ the widely-used gradient accumulation technique~\cite{smith2022using} with an accumulation step of 20 to avoid frequent gradient synchronization. We keep other settings the same as in \autoref{sec:control-single}.

We report in \autoref{fig:aws-control-5min} the training performance under a random node failure of every 5 minutes. \sys outperforms Tutel(FT) and Tutel by 1.8x and 2.1x for GPT-S, and by 3.9x and 4.0x for GPT-L. We also observe that with more nodes in the cluster, Tutel(FT) can leave many unused, as it cannot reuse idle nodes that are previously dropped without restarting, same as DS(FT) in \autoref{sec:control-single}. 
In the supplementary material, we provide the results under an ideal case where subsequent failed nodes are the unused ones by Tutel and Tutel(FT).

\VSPACE{-2mm}
\subsection{Simulation}
We evaluate how \sys scales to larger models and larger clusters via simulations. We follow Bamboo~\cite{thorpe2023bamboo} to setup our simulation, using a constant node preemption probability while varying new node allocation probabilities per simulation hour. We simulate the training of the DeepSeek~V3 model~\cite{liu2024deepseek}, where each node has 8 H200 GPUs and 8 400~Gbps NICs. We use the performance model from \cite{shallowsim}. We consider a mixture of data, tensor and expert parallelism, where tensor parallelism is used within each node for non-expert components, while each GPU holds 4 expert replicas for both traditional expert parallelism and \sys. We sample an expert routing trace for DeepSeek-V3 using ShareGPT dataset~\cite{zheng2023judging}. We run 10 trials for each setting.

In \autoref{fig:simulation_fail_prob}, we show the job completion time (JCT) for training 5M samples with 512 GPUs (64 nodes) under different failure probabilities, where the JCT is normalized to \sys. With a low failure probability of 0.01, \sys outperforms DS(FT) and DS by 2.9x and 3.9x, while it increases to 3.6x and 14.5x under a failure probability of 0.15.
In \autoref{fig:simulation_num_gpus}, we show the JCT under a failure probability of 0.05 with different numbers of GPUs. With 256 GPUs, \sys speeds up over DS(FT) and DS by 2.9x and 4.5x, while with 1K GPUs, \sys' speed-up increases to 3.5x and 10.6x. We also observed that our expert placement algorithm in \autoref{sec:expert-placement} takes less than 1~sec on a single CPU core for 1K GPUs.
\VSPACE{-2mm}

\section{Related Work}

\textbf{MoE training systems}
Extensive studies have focused on optimizing MoE training. A series of works~\cite{rajbhandari2022deepspeed,shi2024schemoe,jiang2024lancet,li2023accelerating,liu2023janus,nie2022hetumoe} optimize the all-to-all communication performance. Another line of works design different MoE algorithms and architectures~\cite{rajbhandari2022deepspeed,nie2022hetumoe,li2024locmoe,zhou2022mixture,zoph2022st,zuo2021taming,chi2022representation}. Various system optimizations have been proposed to deal with the imbalanced workload. For example, Tutel~\cite{hwang2023tutel} and SmartMoE~\cite{zhai2023smartmoe} propose dynamic parallelism switching; FasterMoE~\cite{he2022fastermoe} and FlexMoE~\cite{nie2023flexmoe} also utilize the idea of expert replication. However, these works all focus on speeding up training on a fixed-sized cluster, while \sys considers an elastic environment where resiliency and quick reconfiguration is crucial. Many of these optimizations can also be integrated to \sys.

\noindent
\textbf{Fault-tolerant and elastic training.}
Early efforts in elastic training focus on small models trained with pure data parallelism. TorchElastic~\cite{torchelastic} restarts a job upon node membership changes. Elastically allocating resources among multiple jobs have also been explored in~\cite{qiao2021pollux,zheng2023shockwave,hwang2021elastic,gu2023elasticflow,li2022aryl}. However, they do not work for modern LLMs which are frequently well beyond a single GPU's memory capacity. To enable resilient training of large models, many checkpointing optimization techniques have been proposed~\cite{wang2023reliable,wang2023gemini,cai2025moc,gandhi2024moetion,wan2024bytecheckpoint}. Yet, they lack elasticity and require replacement nodes to resume training.
We note that Gemini~\cite{wang2023gemini} designs a strategy for placing checkpoints in CPU memory across machines to maximize recovery probability. However, it assumes each GPU's checkpoint has the same number of replicas, hence does not apply to our expert placement problem, where different experts have different number of replicas. Systems supporting both resilient and elastic training of LLMs~\cite{thorpe2023bamboo,jang2023oobleck,duan2024parcae} are all based on pipeline parallelism, utilizing its flexibility in stage-device mapping. These works are complementary to \sys, where we target expert parallelism introduced in MoE.

\VSPACE{-2mm}

\section{Conclusion}
This paper presents \sys, the first system for resilient and elastic distributed training of MoE models. \sys adaptively allocates replicas based on the expert routing distribution of the workload to speed-up training. With a proven optimal expert placement strategy, \sys maximizes the probability of failure recovery. Upon failures, \sys efficiently migrates to a new expert placement plan with all remaining GPUs fully utilized. Our results show that \sys outperforms state-of-the-art checkpointing based MoE training systems by up to 5.7x under frequent node failures and 3.4x on a real spot instance trace. We will open source \sys.

\bibliographystyle{ACM-Reference-Format}
\bibliography{reference}


\begin{thebibliography}{46}


\ifx \showCODEN    \undefined \def \showCODEN     #1{\unskip}     \fi
\ifx \showISBNx    \undefined \def \showISBNx     #1{\unskip}     \fi
\ifx \showISBNxiii \undefined \def \showISBNxiii  #1{\unskip}     \fi
\ifx \showISSN     \undefined \def \showISSN      #1{\unskip}     \fi
\ifx \showLCCN     \undefined \def \showLCCN      #1{\unskip}     \fi
\ifx \shownote     \undefined \def \shownote      #1{#1}          \fi
\ifx \showarticletitle \undefined \def \showarticletitle #1{#1}   \fi
\ifx \showURL      \undefined \def \showURL       {\relax}        \fi
\providecommand\bibfield[2]{#2}
\providecommand\bibinfo[2]{#2}
\providecommand\natexlab[1]{#1}
\providecommand\showeprint[2][]{arXiv:#2}

\bibitem[sha(2025)]%
        {shallowsim}
 \bibinfo{year}{2025}\natexlab{}.
\newblock \bibinfo{title}{{DeepSeek-V3/R1 Performance Simulator}}.
\newblock \bibinfo{howpublished}{\url{https://github.com/zartbot/shallowsim}}.
\newblock


\bibitem[Cai et~al\mbox{.}(2025)]%
        {cai2025moc}
\bibfield{author}{\bibinfo{person}{Weilin Cai}, \bibinfo{person}{Le Qin}, {and} \bibinfo{person}{Jiayi Huang}.} \bibinfo{year}{2025}\natexlab{}.
\newblock \showarticletitle{MoC-System: Efficient Fault Tolerance for Sparse Mixture-of-Experts Model Training}. In \bibinfo{booktitle}{\emph{Proceedings of the 30th ACM International Conference on Architectural Support for Programming Languages and Operating Systems, Volume 2}}. \bibinfo{pages}{655--671}.
\newblock


\bibitem[Chi et~al\mbox{.}(2022)]%
        {chi2022representation}
\bibfield{author}{\bibinfo{person}{Zewen Chi}, \bibinfo{person}{Li Dong}, \bibinfo{person}{Shaohan Huang}, \bibinfo{person}{Damai Dai}, \bibinfo{person}{Shuming Ma}, \bibinfo{person}{Barun Patra}, \bibinfo{person}{Saksham Singhal}, \bibinfo{person}{Payal Bajaj}, \bibinfo{person}{Xia Song}, \bibinfo{person}{Xian-Ling Mao}, {et~al\mbox{.}}} \bibinfo{year}{2022}\natexlab{}.
\newblock \showarticletitle{On the representation collapse of sparse mixture of experts}.
\newblock \bibinfo{journal}{\emph{Advances in Neural Information Processing Systems}}  \bibinfo{volume}{35} (\bibinfo{year}{2022}), \bibinfo{pages}{34600--34613}.
\newblock


\bibitem[Duan et~al\mbox{.}(2024)]%
        {duan2024parcae}
\bibfield{author}{\bibinfo{person}{Jiangfei Duan}, \bibinfo{person}{Ziang Song}, \bibinfo{person}{Xupeng Miao}, \bibinfo{person}{Xiaoli Xi}, \bibinfo{person}{Dahua Lin}, \bibinfo{person}{Harry Xu}, \bibinfo{person}{Minjia Zhang}, {and} \bibinfo{person}{Zhihao Jia}.} \bibinfo{year}{2024}\natexlab{}.
\newblock \showarticletitle{Parcae: Proactive,$\{$Liveput-Optimized$\}$$\{$DNN$\}$ Training on Preemptible Instances}. In \bibinfo{booktitle}{\emph{21st USENIX Symposium on Networked Systems Design and Implementation (NSDI 24)}}. \bibinfo{pages}{1121--1139}.
\newblock


\bibitem[Fedus et~al\mbox{.}(2022)]%
        {fedus2022switch}
\bibfield{author}{\bibinfo{person}{William Fedus}, \bibinfo{person}{Barret Zoph}, {and} \bibinfo{person}{Noam Shazeer}.} \bibinfo{year}{2022}\natexlab{}.
\newblock \showarticletitle{Switch transformers: Scaling to trillion parameter models with simple and efficient sparsity}.
\newblock \bibinfo{journal}{\emph{Journal of Machine Learning Research}} \bibinfo{volume}{23}, \bibinfo{number}{120} (\bibinfo{year}{2022}), \bibinfo{pages}{1--39}.
\newblock


\bibitem[Gandhi and Kozyrakis(2024)]%
        {gandhi2024moetion}
\bibfield{author}{\bibinfo{person}{Swapnil Gandhi} {and} \bibinfo{person}{Christos Kozyrakis}.} \bibinfo{year}{2024}\natexlab{}.
\newblock \showarticletitle{MoEtion: Efficient and Reliable Checkpointing for Mixture-of-Experts Models at Scale}.
\newblock \bibinfo{journal}{\emph{arXiv preprint arXiv:2412.15411}} (\bibinfo{year}{2024}).
\newblock


\bibitem[Gu et~al\mbox{.}(2023)]%
        {gu2023elasticflow}
\bibfield{author}{\bibinfo{person}{Diandian Gu}, \bibinfo{person}{Yihao Zhao}, \bibinfo{person}{Yinmin Zhong}, \bibinfo{person}{Yifan Xiong}, \bibinfo{person}{Zhenhua Han}, \bibinfo{person}{Peng Cheng}, \bibinfo{person}{Fan Yang}, \bibinfo{person}{Gang Huang}, \bibinfo{person}{Xin Jin}, {and} \bibinfo{person}{Xuanzhe Liu}.} \bibinfo{year}{2023}\natexlab{}.
\newblock \showarticletitle{ElasticFlow: An elastic serverless training platform for distributed deep learning}. In \bibinfo{booktitle}{\emph{Proceedings of the 28th ACM International Conference on Architectural Support for Programming Languages and Operating Systems, Volume 2}}. \bibinfo{pages}{266--280}.
\newblock


\bibitem[He et~al\mbox{.}(2022)]%
        {he2022fastermoe}
\bibfield{author}{\bibinfo{person}{Jiaao He}, \bibinfo{person}{Jidong Zhai}, \bibinfo{person}{Tiago Antunes}, \bibinfo{person}{Haojie Wang}, \bibinfo{person}{Fuwen Luo}, \bibinfo{person}{Shangfeng Shi}, {and} \bibinfo{person}{Qin Li}.} \bibinfo{year}{2022}\natexlab{}.
\newblock \showarticletitle{Fastermoe: modeling and optimizing training of large-scale dynamic pre-trained models}. In \bibinfo{booktitle}{\emph{Proceedings of the 27th ACM SIGPLAN Symposium on Principles and Practice of Parallel Programming}}. \bibinfo{pages}{120--134}.
\newblock


\bibitem[He et~al\mbox{.}(2023)]%
        {he2023unicron}
\bibfield{author}{\bibinfo{person}{Tao He}, \bibinfo{person}{Xue Li}, \bibinfo{person}{Zhibin Wang}, \bibinfo{person}{Kun Qian}, \bibinfo{person}{Jingbo Xu}, \bibinfo{person}{Wenyuan Yu}, {and} \bibinfo{person}{Jingren Zhou}.} \bibinfo{year}{2023}\natexlab{}.
\newblock \showarticletitle{Unicron: Economizing self-healing llm training at scale}.
\newblock \bibinfo{journal}{\emph{arXiv preprint arXiv:2401.00134}} (\bibinfo{year}{2023}).
\newblock


\bibitem[Hwang et~al\mbox{.}(2023)]%
        {hwang2023tutel}
\bibfield{author}{\bibinfo{person}{Changho Hwang}, \bibinfo{person}{Wei Cui}, \bibinfo{person}{Yifan Xiong}, \bibinfo{person}{Ziyue Yang}, \bibinfo{person}{Ze Liu}, \bibinfo{person}{Han Hu}, \bibinfo{person}{Zilong Wang}, \bibinfo{person}{Rafael Salas}, \bibinfo{person}{Jithin Jose}, \bibinfo{person}{Prabhat Ram}, {et~al\mbox{.}}} \bibinfo{year}{2023}\natexlab{}.
\newblock \showarticletitle{Tutel: Adaptive mixture-of-experts at scale}.
\newblock \bibinfo{journal}{\emph{Proceedings of Machine Learning and Systems}}  \bibinfo{volume}{5} (\bibinfo{year}{2023}).
\newblock


\bibitem[Hwang et~al\mbox{.}(2021)]%
        {hwang2021elastic}
\bibfield{author}{\bibinfo{person}{Changho Hwang}, \bibinfo{person}{Taehyun Kim}, \bibinfo{person}{Sunghyun Kim}, \bibinfo{person}{Jinwoo Shin}, {and} \bibinfo{person}{KyoungSoo Park}.} \bibinfo{year}{2021}\natexlab{}.
\newblock \showarticletitle{Elastic resource sharing for distributed deep learning}. In \bibinfo{booktitle}{\emph{18th USENIX Symposium on Networked Systems Design and Implementation (NSDI 21)}}. \bibinfo{pages}{721--739}.
\newblock


\bibitem[Imambi et~al\mbox{.}(2021)]%
        {imambi2021pytorch}
\bibfield{author}{\bibinfo{person}{Sagar Imambi}, \bibinfo{person}{Kolla~Bhanu Prakash}, {and} \bibinfo{person}{GR Kanagachidambaresan}.} \bibinfo{year}{2021}\natexlab{}.
\newblock \showarticletitle{PyTorch}.
\newblock \bibinfo{journal}{\emph{Programming with TensorFlow: Solution for Edge Computing Applications}} (\bibinfo{year}{2021}), \bibinfo{pages}{87--104}.
\newblock


\bibitem[Jang et~al\mbox{.}(2023)]%
        {jang2023oobleck}
\bibfield{author}{\bibinfo{person}{Insu Jang}, \bibinfo{person}{Zhenning Yang}, \bibinfo{person}{Zhen Zhang}, \bibinfo{person}{Xin Jin}, {and} \bibinfo{person}{Mosharaf Chowdhury}.} \bibinfo{year}{2023}\natexlab{}.
\newblock \showarticletitle{Oobleck: Resilient distributed training of large models using pipeline templates}. In \bibinfo{booktitle}{\emph{Proceedings of the 29th Symposium on Operating Systems Principles}}. \bibinfo{pages}{382--395}.
\newblock


\bibitem[Jiang et~al\mbox{.}(2024a)]%
        {jiang2024mixtral}
\bibfield{author}{\bibinfo{person}{Albert~Q Jiang}, \bibinfo{person}{Alexandre Sablayrolles}, \bibinfo{person}{Antoine Roux}, \bibinfo{person}{Arthur Mensch}, \bibinfo{person}{Blanche Savary}, \bibinfo{person}{Chris Bamford}, \bibinfo{person}{Devendra~Singh Chaplot}, \bibinfo{person}{Diego de~las Casas}, \bibinfo{person}{Emma~Bou Hanna}, \bibinfo{person}{Florian Bressand}, {et~al\mbox{.}}} \bibinfo{year}{2024}\natexlab{a}.
\newblock \showarticletitle{Mixtral of experts}.
\newblock \bibinfo{journal}{\emph{arXiv preprint arXiv:2401.04088}} (\bibinfo{year}{2024}).
\newblock


\bibitem[Jiang et~al\mbox{.}(2024b)]%
        {jiang2024lancet}
\bibfield{author}{\bibinfo{person}{Chenyu Jiang}, \bibinfo{person}{Ye Tian}, \bibinfo{person}{Zhen Jia}, \bibinfo{person}{Shuai Zheng}, \bibinfo{person}{Chuan Wu}, {and} \bibinfo{person}{Yida Wang}.} \bibinfo{year}{2024}\natexlab{b}.
\newblock \showarticletitle{Lancet: Accelerating Mixture-of-Experts Training via Whole Graph Computation-Communication Overlapping}.
\newblock \bibinfo{journal}{\emph{arXiv preprint arXiv:2404.19429}} (\bibinfo{year}{2024}).
\newblock


\bibitem[Kokolis et~al\mbox{.}(2025)]%
        {kokolis2025revisiting}
\bibfield{author}{\bibinfo{person}{Apostolos Kokolis}, \bibinfo{person}{Michael Kuchnik}, \bibinfo{person}{John Hoffman}, \bibinfo{person}{Adithya Kumar}, \bibinfo{person}{Parth Malani}, \bibinfo{person}{Faye Ma}, \bibinfo{person}{Zachary DeVito}, \bibinfo{person}{Shubho Sengupta}, \bibinfo{person}{Kalyan Saladi}, {and} \bibinfo{person}{Carole-Jean Wu}.} \bibinfo{year}{2025}\natexlab{}.
\newblock \showarticletitle{Revisiting Reliability in Large-Scale Machine Learning Research Clusters}. In \bibinfo{booktitle}{\emph{2025 IEEE International Symposium on High Performance Computer Architecture (HPCA)}}. IEEE, \bibinfo{pages}{1259--1274}.
\newblock


\bibitem[Lepikhin et~al\mbox{.}(2020)]%
        {lepikhin2020gshard}
\bibfield{author}{\bibinfo{person}{Dmitry Lepikhin}, \bibinfo{person}{HyoukJoong Lee}, \bibinfo{person}{Yuanzhong Xu}, \bibinfo{person}{Dehao Chen}, \bibinfo{person}{Orhan Firat}, \bibinfo{person}{Yanping Huang}, \bibinfo{person}{Maxim Krikun}, \bibinfo{person}{Noam Shazeer}, {and} \bibinfo{person}{Zhifeng Chen}.} \bibinfo{year}{2020}\natexlab{}.
\newblock \showarticletitle{Gshard: Scaling giant models with conditional computation and automatic sharding}.
\newblock \bibinfo{journal}{\emph{arXiv preprint arXiv:2006.16668}} (\bibinfo{year}{2020}).
\newblock


\bibitem[Li et~al\mbox{.}(2023)]%
        {li2023accelerating}
\bibfield{author}{\bibinfo{person}{Jiamin Li}, \bibinfo{person}{Yimin Jiang}, \bibinfo{person}{Yibo Zhu}, \bibinfo{person}{Cong Wang}, {and} \bibinfo{person}{Hong Xu}.} \bibinfo{year}{2023}\natexlab{}.
\newblock \showarticletitle{Accelerating distributed $\{$MoE$\}$ training and inference with lina}. In \bibinfo{booktitle}{\emph{2023 USENIX Annual Technical Conference (USENIX ATC 23)}}. \bibinfo{pages}{945--959}.
\newblock


\bibitem[Li et~al\mbox{.}(2024)]%
        {li2024locmoe}
\bibfield{author}{\bibinfo{person}{Jing Li}, \bibinfo{person}{Zhijie Sun}, \bibinfo{person}{Xuan He}, \bibinfo{person}{Li Zeng}, \bibinfo{person}{Yi Lin}, \bibinfo{person}{Entong Li}, \bibinfo{person}{Binfan Zheng}, \bibinfo{person}{Rongqian Zhao}, {and} \bibinfo{person}{Xin Chen}.} \bibinfo{year}{2024}\natexlab{}.
\newblock \showarticletitle{Locmoe: A low-overhead moe for large language model training}.
\newblock \bibinfo{journal}{\emph{arXiv preprint arXiv:2401.13920}} (\bibinfo{year}{2024}).
\newblock


\bibitem[Li et~al\mbox{.}(2022)]%
        {li2022aryl}
\bibfield{author}{\bibinfo{person}{Jiamin Li}, \bibinfo{person}{Hong Xu}, \bibinfo{person}{Yibo Zhu}, \bibinfo{person}{Zherui Liu}, \bibinfo{person}{Chuanxiong Guo}, {and} \bibinfo{person}{Cong Wang}.} \bibinfo{year}{2022}\natexlab{}.
\newblock \showarticletitle{Aryl: An elastic cluster scheduler for deep learning}.
\newblock \bibinfo{journal}{\emph{arXiv preprint arXiv:2202.07896}} (\bibinfo{year}{2022}).
\newblock


\bibitem[Liu et~al\mbox{.}(2024)]%
        {liu2024deepseek}
\bibfield{author}{\bibinfo{person}{Aixin Liu}, \bibinfo{person}{Bei Feng}, \bibinfo{person}{Bing Xue}, \bibinfo{person}{Bingxuan Wang}, \bibinfo{person}{Bochao Wu}, \bibinfo{person}{Chengda Lu}, \bibinfo{person}{Chenggang Zhao}, \bibinfo{person}{Chengqi Deng}, \bibinfo{person}{Chenyu Zhang}, \bibinfo{person}{Chong Ruan}, {et~al\mbox{.}}} \bibinfo{year}{2024}\natexlab{}.
\newblock \showarticletitle{Deepseek-v3 technical report}.
\newblock \bibinfo{journal}{\emph{arXiv preprint arXiv:2412.19437}} (\bibinfo{year}{2024}).
\newblock


\bibitem[Liu et~al\mbox{.}(2023)]%
        {liu2023janus}
\bibfield{author}{\bibinfo{person}{Juncai Liu}, \bibinfo{person}{Jessie~Hui Wang}, {and} \bibinfo{person}{Yimin Jiang}.} \bibinfo{year}{2023}\natexlab{}.
\newblock \showarticletitle{Janus: A unified distributed training framework for sparse mixture-of-experts models}. In \bibinfo{booktitle}{\emph{Proceedings of the ACM SIGCOMM 2023 Conference}}. \bibinfo{pages}{486--498}.
\newblock


\bibitem[llama4(2024)]%
        {llama4}
llama4 \bibinfo{year}{2024}\natexlab{}.
\newblock \bibinfo{title}{{Meta Llama 4}}.
\newblock \bibinfo{howpublished}{\url{https://ai.meta.com/blog/llama-4-multimodal-intelligence/}}.
\newblock


\bibitem[Maeng et~al\mbox{.}(2021)]%
        {maeng2021understanding}
\bibfield{author}{\bibinfo{person}{Kiwan Maeng}, \bibinfo{person}{Shivam Bharuka}, \bibinfo{person}{Isabel Gao}, \bibinfo{person}{Mark Jeffrey}, \bibinfo{person}{Vikram Saraph}, \bibinfo{person}{Bor-Yiing Su}, \bibinfo{person}{Caroline Trippel}, \bibinfo{person}{Jiyan Yang}, \bibinfo{person}{Mike Rabbat}, \bibinfo{person}{Brandon Lucia}, {et~al\mbox{.}}} \bibinfo{year}{2021}\natexlab{}.
\newblock \showarticletitle{Understanding and improving failure tolerant training for deep learning recommendation with partial recovery}.
\newblock \bibinfo{journal}{\emph{Proceedings of Machine Learning and Systems}}  \bibinfo{volume}{3} (\bibinfo{year}{2021}), \bibinfo{pages}{637--651}.
\newblock


\bibitem[Merity et~al\mbox{.}(2016)]%
        {merity2016pointer}
\bibfield{author}{\bibinfo{person}{Stephen Merity}, \bibinfo{person}{Caiming Xiong}, \bibinfo{person}{James Bradbury}, {and} \bibinfo{person}{Richard Socher}.} \bibinfo{year}{2016}\natexlab{}.
\newblock \showarticletitle{Pointer sentinel mixture models}.
\newblock \bibinfo{journal}{\emph{arXiv preprint arXiv:1609.07843}} (\bibinfo{year}{2016}).
\newblock


\bibitem[nccl(2024)]%
        {nccl}
nccl \bibinfo{year}{2024}\natexlab{}.
\newblock \bibinfo{title}{{The NVIDIA Collective Communication Library (NCCL)}}.
\newblock \bibinfo{howpublished}{\url{https://developer.nvidia.com/nccl}}.
\newblock


\bibitem[Nie et~al\mbox{.}(2023)]%
        {nie2023flexmoe}
\bibfield{author}{\bibinfo{person}{Xiaonan Nie}, \bibinfo{person}{Xupeng Miao}, \bibinfo{person}{Zilong Wang}, \bibinfo{person}{Zichao Yang}, \bibinfo{person}{Jilong Xue}, \bibinfo{person}{Lingxiao Ma}, \bibinfo{person}{Gang Cao}, {and} \bibinfo{person}{Bin Cui}.} \bibinfo{year}{2023}\natexlab{}.
\newblock \showarticletitle{Flexmoe: Scaling large-scale sparse pre-trained model training via dynamic device placement}.
\newblock \bibinfo{journal}{\emph{Proceedings of the ACM on Management of Data}} \bibinfo{volume}{1}, \bibinfo{number}{1} (\bibinfo{year}{2023}), \bibinfo{pages}{1--19}.
\newblock


\bibitem[Nie et~al\mbox{.}(2022)]%
        {nie2022hetumoe}
\bibfield{author}{\bibinfo{person}{Xiaonan Nie}, \bibinfo{person}{Pinxue Zhao}, \bibinfo{person}{Xupeng Miao}, \bibinfo{person}{Tong Zhao}, {and} \bibinfo{person}{Bin Cui}.} \bibinfo{year}{2022}\natexlab{}.
\newblock \showarticletitle{HetuMoE: An efficient trillion-scale mixture-of-expert distributed training system}.
\newblock \bibinfo{journal}{\emph{arXiv preprint arXiv:2203.14685}} (\bibinfo{year}{2022}).
\newblock


\bibitem[Qiao et~al\mbox{.}(2021)]%
        {qiao2021pollux}
\bibfield{author}{\bibinfo{person}{Aurick Qiao}, \bibinfo{person}{Sang~Keun Choe}, \bibinfo{person}{Suhas~Jayaram Subramanya}, \bibinfo{person}{Willie Neiswanger}, \bibinfo{person}{Qirong Ho}, \bibinfo{person}{Hao Zhang}, \bibinfo{person}{Gregory~R Ganger}, {and} \bibinfo{person}{Eric~P Xing}.} \bibinfo{year}{2021}\natexlab{}.
\newblock \showarticletitle{Pollux: Co-adaptive cluster scheduling for goodput-optimized deep learning}. In \bibinfo{booktitle}{\emph{15th $\{$USENIX$\}$ Symposium on Operating Systems Design and Implementation ($\{$OSDI$\}$ 21)}}.
\newblock


\bibitem[qwen3(2024)]%
        {qwen3}
qwen3 \bibinfo{year}{2024}\natexlab{}.
\newblock \bibinfo{title}{{Alibaba Qwen3}}.
\newblock \bibinfo{howpublished}{\url{https://github.com/QwenLM/Qwen3}}.
\newblock


\bibitem[Rajbhandari et~al\mbox{.}(2022)]%
        {rajbhandari2022deepspeed}
\bibfield{author}{\bibinfo{person}{Samyam Rajbhandari}, \bibinfo{person}{Conglong Li}, \bibinfo{person}{Zhewei Yao}, \bibinfo{person}{Minjia Zhang}, \bibinfo{person}{Reza~Yazdani Aminabadi}, \bibinfo{person}{Ammar~Ahmad Awan}, \bibinfo{person}{Jeff Rasley}, {and} \bibinfo{person}{Yuxiong He}.} \bibinfo{year}{2022}\natexlab{}.
\newblock \showarticletitle{Deepspeed-moe: Advancing mixture-of-experts inference and training to power next-generation ai scale}. In \bibinfo{booktitle}{\emph{International conference on machine learning}}. PMLR, \bibinfo{pages}{18332--18346}.
\newblock


\bibitem[Rasley et~al\mbox{.}(2020)]%
        {rasley2020deepspeed}
\bibfield{author}{\bibinfo{person}{Jeff Rasley}, \bibinfo{person}{Samyam Rajbhandari}, \bibinfo{person}{Olatunji Ruwase}, {and} \bibinfo{person}{Yuxiong He}.} \bibinfo{year}{2020}\natexlab{}.
\newblock \showarticletitle{Deepspeed: System optimizations enable training deep learning models with over 100 billion parameters}. In \bibinfo{booktitle}{\emph{Proceedings of the 26th ACM SIGKDD International Conference on Knowledge Discovery \& Data Mining}}. \bibinfo{pages}{3505--3506}.
\newblock


\bibitem[Shi et~al\mbox{.}(2024)]%
        {shi2024schemoe}
\bibfield{author}{\bibinfo{person}{Shaohuai Shi}, \bibinfo{person}{Xinglin Pan}, \bibinfo{person}{Qiang Wang}, \bibinfo{person}{Chengjian Liu}, \bibinfo{person}{Xiaozhe Ren}, \bibinfo{person}{Zhongzhe Hu}, \bibinfo{person}{Yu Yang}, \bibinfo{person}{Bo Li}, {and} \bibinfo{person}{Xiaowen Chu}.} \bibinfo{year}{2024}\natexlab{}.
\newblock \showarticletitle{ScheMoE: An Extensible Mixture-of-Experts Distributed Training System with Tasks Scheduling}. In \bibinfo{booktitle}{\emph{Proceedings of the Nineteenth European Conference on Computer Systems}}. \bibinfo{pages}{236--249}.
\newblock


\bibitem[Smith et~al\mbox{.}(2022)]%
        {smith2022using}
\bibfield{author}{\bibinfo{person}{Shaden Smith}, \bibinfo{person}{Mostofa Patwary}, \bibinfo{person}{Brandon Norick}, \bibinfo{person}{Patrick LeGresley}, \bibinfo{person}{Samyam Rajbhandari}, \bibinfo{person}{Jared Casper}, \bibinfo{person}{Zhun Liu}, \bibinfo{person}{Shrimai Prabhumoye}, \bibinfo{person}{George Zerveas}, \bibinfo{person}{Vijay Korthikanti}, {et~al\mbox{.}}} \bibinfo{year}{2022}\natexlab{}.
\newblock \showarticletitle{Using deepspeed and megatron to train megatron-turing nlg 530b, a large-scale generative language model}.
\newblock \bibinfo{journal}{\emph{arXiv preprint arXiv:2201.11990}} (\bibinfo{year}{2022}).
\newblock


\bibitem[Thorpe et~al\mbox{.}(2023)]%
        {thorpe2023bamboo}
\bibfield{author}{\bibinfo{person}{John Thorpe}, \bibinfo{person}{Pengzhan Zhao}, \bibinfo{person}{Jonathan Eyolfson}, \bibinfo{person}{Yifan Qiao}, \bibinfo{person}{Zhihao Jia}, \bibinfo{person}{Minjia Zhang}, \bibinfo{person}{Ravi Netravali}, {and} \bibinfo{person}{Guoqing~Harry Xu}.} \bibinfo{year}{2023}\natexlab{}.
\newblock \showarticletitle{Bamboo: Making preemptible instances resilient for affordable training of large $\{$DNNs$\}$}. In \bibinfo{booktitle}{\emph{20th USENIX Symposium on Networked Systems Design and Implementation (NSDI 23)}}. \bibinfo{pages}{497--513}.
\newblock


\bibitem[torchelastic(2024)]%
        {torchelastic}
torchelastic \bibinfo{year}{2024}\natexlab{}.
\newblock \bibinfo{title}{{TorchElastic}}.
\newblock \bibinfo{howpublished}{\url{https://pytorch.org/docs/stable/distributed.elastic.html}}.
\newblock


\bibitem[Wan et~al\mbox{.}(2024)]%
        {wan2024bytecheckpoint}
\bibfield{author}{\bibinfo{person}{Borui Wan}, \bibinfo{person}{Mingji Han}, \bibinfo{person}{Yiyao Sheng}, \bibinfo{person}{Yanghua Peng}, \bibinfo{person}{Haibin Lin}, \bibinfo{person}{Mofan Zhang}, \bibinfo{person}{Zhichao Lai}, \bibinfo{person}{Menghan Yu}, \bibinfo{person}{Junda Zhang}, \bibinfo{person}{Zuquan Song}, {et~al\mbox{.}}} \bibinfo{year}{2024}\natexlab{}.
\newblock \showarticletitle{ByteCheckpoint: A Unified Checkpointing System for Large Foundation Model Development}.
\newblock \bibinfo{journal}{\emph{arXiv preprint arXiv:2407.20143}} (\bibinfo{year}{2024}).
\newblock


\bibitem[Wang et~al\mbox{.}(2023b)]%
        {wang2023reliable}
\bibfield{author}{\bibinfo{person}{Yuxin Wang}, \bibinfo{person}{Shaohuai Shi}, \bibinfo{person}{Xin He}, \bibinfo{person}{Zhenheng Tang}, \bibinfo{person}{Xinglin Pan}, \bibinfo{person}{Yang Zheng}, \bibinfo{person}{Xiaoyu Wu}, \bibinfo{person}{Amelie~Chi Zhou}, \bibinfo{person}{Bingsheng He}, {and} \bibinfo{person}{Xiaowen Chu}.} \bibinfo{year}{2023}\natexlab{b}.
\newblock \showarticletitle{Reliable and Efficient In-Memory Fault Tolerance of Large Language Model Pretraining}.
\newblock \bibinfo{journal}{\emph{arXiv preprint arXiv:2310.12670}} (\bibinfo{year}{2023}).
\newblock


\bibitem[Wang et~al\mbox{.}(2023a)]%
        {wang2023gemini}
\bibfield{author}{\bibinfo{person}{Zhuang Wang}, \bibinfo{person}{Zhen Jia}, \bibinfo{person}{Shuai Zheng}, \bibinfo{person}{Zhen Zhang}, \bibinfo{person}{Xinwei Fu}, \bibinfo{person}{TS~Eugene Ng}, {and} \bibinfo{person}{Yida Wang}.} \bibinfo{year}{2023}\natexlab{a}.
\newblock \showarticletitle{Gemini: Fast failure recovery in distributed training with in-memory checkpoints}. In \bibinfo{booktitle}{\emph{Proceedings of the 29th Symposium on Operating Systems Principles}}. \bibinfo{pages}{364--381}.
\newblock


\bibitem[Zhai et~al\mbox{.}(2023)]%
        {zhai2023smartmoe}
\bibfield{author}{\bibinfo{person}{Mingshu Zhai}, \bibinfo{person}{Jiaao He}, \bibinfo{person}{Zixuan Ma}, \bibinfo{person}{Zan Zong}, \bibinfo{person}{Runqing Zhang}, {and} \bibinfo{person}{Jidong Zhai}.} \bibinfo{year}{2023}\natexlab{}.
\newblock \showarticletitle{$\{$SmartMoE$\}$: Efficiently Training $\{$Sparsely-Activated$\}$ Models through Combining Offline and Online Parallelization}. In \bibinfo{booktitle}{\emph{2023 USENIX Annual Technical Conference (USENIX ATC 23)}}. \bibinfo{pages}{961--975}.
\newblock


\bibitem[Zhang et~al\mbox{.}(2022)]%
        {zhang2022opt}
\bibfield{author}{\bibinfo{person}{Susan Zhang}, \bibinfo{person}{Stephen Roller}, \bibinfo{person}{Naman Goyal}, \bibinfo{person}{Mikel Artetxe}, \bibinfo{person}{Moya Chen}, \bibinfo{person}{Shuohui Chen}, \bibinfo{person}{Christopher Dewan}, \bibinfo{person}{Mona Diab}, \bibinfo{person}{Xian Li}, \bibinfo{person}{Xi~Victoria Lin}, {et~al\mbox{.}}} \bibinfo{year}{2022}\natexlab{}.
\newblock \showarticletitle{Opt: Open pre-trained transformer language models}.
\newblock \bibinfo{journal}{\emph{arXiv preprint arXiv:2205.01068}} (\bibinfo{year}{2022}).
\newblock


\bibitem[Zheng et~al\mbox{.}(2023a)]%
        {zheng2023judging}
\bibfield{author}{\bibinfo{person}{Lianmin Zheng}, \bibinfo{person}{Wei-Lin Chiang}, \bibinfo{person}{Ying Sheng}, \bibinfo{person}{Siyuan Zhuang}, \bibinfo{person}{Zhanghao Wu}, \bibinfo{person}{Yonghao Zhuang}, \bibinfo{person}{Zi Lin}, \bibinfo{person}{Zhuohan Li}, \bibinfo{person}{Dacheng Li}, \bibinfo{person}{Eric.~P Xing}, \bibinfo{person}{Hao Zhang}, \bibinfo{person}{Joseph~E. Gonzalez}, {and} \bibinfo{person}{Ion Stoica}.} \bibinfo{year}{2023}\natexlab{a}.
\newblock \bibinfo{title}{Judging LLM-as-a-judge with MT-Bench and Chatbot Arena}.
\newblock
\showeprint[arxiv]{2306.05685}~[cs.CL]


\bibitem[Zheng et~al\mbox{.}(2023b)]%
        {zheng2023shockwave}
\bibfield{author}{\bibinfo{person}{Pengfei Zheng}, \bibinfo{person}{Rui Pan}, \bibinfo{person}{Tarannum Khan}, \bibinfo{person}{Shivaram Venkataraman}, {and} \bibinfo{person}{Aditya Akella}.} \bibinfo{year}{2023}\natexlab{b}.
\newblock \showarticletitle{Shockwave: Fair and efficient cluster scheduling for dynamic adaptation in machine learning}. In \bibinfo{booktitle}{\emph{20th USENIX Symposium on Networked Systems Design and Implementation (NSDI 23)}}. \bibinfo{pages}{703--723}.
\newblock


\bibitem[Zhou et~al\mbox{.}(2022)]%
        {zhou2022mixture}
\bibfield{author}{\bibinfo{person}{Yanqi Zhou}, \bibinfo{person}{Tao Lei}, \bibinfo{person}{Hanxiao Liu}, \bibinfo{person}{Nan Du}, \bibinfo{person}{Yanping Huang}, \bibinfo{person}{Vincent Zhao}, \bibinfo{person}{Andrew~M Dai}, \bibinfo{person}{Quoc~V Le}, \bibinfo{person}{James Laudon}, {et~al\mbox{.}}} \bibinfo{year}{2022}\natexlab{}.
\newblock \showarticletitle{Mixture-of-experts with expert choice routing}.
\newblock \bibinfo{journal}{\emph{Advances in Neural Information Processing Systems}}  \bibinfo{volume}{35} (\bibinfo{year}{2022}), \bibinfo{pages}{7103--7114}.
\newblock


\bibitem[Zoph et~al\mbox{.}(2022)]%
        {zoph2022st}
\bibfield{author}{\bibinfo{person}{Barret Zoph}, \bibinfo{person}{Irwan Bello}, \bibinfo{person}{Sameer Kumar}, \bibinfo{person}{Nan Du}, \bibinfo{person}{Yanping Huang}, \bibinfo{person}{Jeff Dean}, \bibinfo{person}{Noam Shazeer}, {and} \bibinfo{person}{William Fedus}.} \bibinfo{year}{2022}\natexlab{}.
\newblock \showarticletitle{St-moe: Designing stable and transferable sparse expert models}.
\newblock \bibinfo{journal}{\emph{arXiv preprint arXiv:2202.08906}} (\bibinfo{year}{2022}).
\newblock


\bibitem[Zuo et~al\mbox{.}(2021)]%
        {zuo2021taming}
\bibfield{author}{\bibinfo{person}{Simiao Zuo}, \bibinfo{person}{Xiaodong Liu}, \bibinfo{person}{Jian Jiao}, \bibinfo{person}{Young~Jin Kim}, \bibinfo{person}{Hany Hassan}, \bibinfo{person}{Ruofei Zhang}, \bibinfo{person}{Tuo Zhao}, {and} \bibinfo{person}{Jianfeng Gao}.} \bibinfo{year}{2021}\natexlab{}.
\newblock \showarticletitle{Taming sparsely activated transformer with stochastic experts}.
\newblock \bibinfo{journal}{\emph{arXiv preprint arXiv:2110.04260}} (\bibinfo{year}{2021}).
\newblock


\end{thebibliography}

\clearpage
\appendix
\section{Supplementary Material}

\begin{figure*}
\centering
\includegraphics[width=0.99\linewidth]{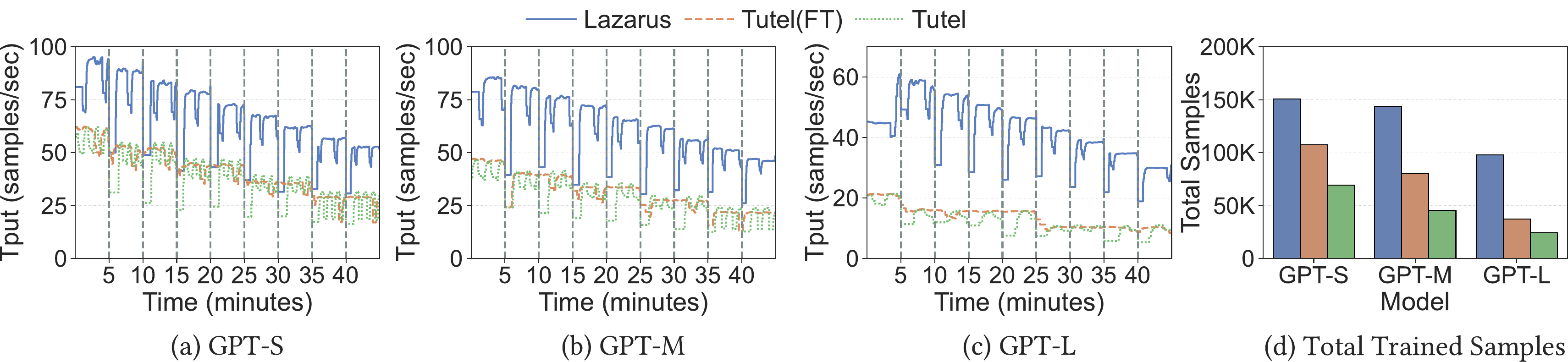}
\caption{\textbf{[Comparison with Tutel]:} Throughput and total trained samples with a single node fails every 5 minutes, where subsequent failed nodes are unused nodes for Tutel and Tutel(FT).}
\label{fig:aws-pairfail-control-5min}
\end{figure*}

\subsection{Training Performance under Ideal Failure Patterns for Tutel and Tutel(FT)}

In \autoref{fig:aws-pairfail-control-5min} we present the training performance under a single node failure every 5 minutes. We ensure that subsequent failed nodes are nodes that are previously dropped by Tutel and Tutel(FT), as the total number of nodes is not a multiple of EP size. We keep other settings the same as in \autoref{sec:comparison_tutel}. 

In this case, after initial failures at 5 minutes, Tutel and Tutel(FT) essentially only encounter a failure every 10 minutes for GPT-S and GPT-M, due to an EP size of 2; every 20 minutes for GPT-L, due to an EP size of 4. In terms of total trained samples, \sys outperforms Tutel(FT) by 1.4x for GPT-S and by 2.6x for GPT-L.

\subsection{Proof of Optimality of the MRO Placement Plan}
Recall the setting of our placement problem, we have $N$ nodes, $E$ experts, each node can hold $c$ expert replicas. The $i-$th expert has $r_i$ replicas. Assume there are $R$ nodes alive simultaneously, we want to find a placement plan that maximizes the probability of recovering all the experts when the $R$ alive nodes are sampled uniformly. 
We denote $[k]$ as the set of $\{1,2,\cdots,k\}$.
We use integer matrix $T\in \mathbb{N}^{c\times N}$ to denote the placement plan, $T_{ij}$ represents the expert placed at node $j$'s $i$-th slot. $T$ satisfies the following properties:\\
\begin{align}
\begin{aligned}
    & T_{ij}\in [E], \forall i\in [c], j \in [N]  \\
    & r_k=\sum_{i=1}^c \sum_{j=1}^N \mathbbm{1}_{T_{ij}=k}  , \forall k \in [E]
\end{aligned}
\end{align}
Without loss of generality, we assume $r$ is sorted in the ascending order, $r_1\le r_2\le\cdots \le r_m$. Let $Col_j$ denote the set composed of elements in the $j$-th column of $T$(removing duplicates), $j=1,\cdots,N$. Let $A$ be the set of $R$ random columns that are alive, $A$ is uniformly sampled. Our goal is:\\
\begin{align}
\begin{aligned}
    & \max \mathsf{Pr}(\bigcup_{a \in A}Col_a=[E])  \\
\end{aligned}
\end{align}

\begin{theorem}
 The \textbf{maximum rank overlap} placement plan (MRO plan) is defined as follows:  $[N]$ could be partitioned into $\lceil \frac{E}{c}\rceil$ disjoint subsets: $|S_i|=r_{1+(i-1)*c},  i \in [\lceil\frac{E}{c}\rceil-1] $, $|S_{\lceil\frac{E}{c}\rceil}|=\min\{N-\sum_{j=1}^{\lceil\frac{E}{c}\rceil-1} r_{1+(\lceil\frac{E}{c}\rceil-1)*c}, r_{1+(\lceil\frac{E}{c}\rceil-1)*c}\}$, such that, for $\forall i \in [\lceil\frac{E}{c}\rceil], j\in S_i$, $\{1+(i-1)*c,\cdots, \min\{i*c, E\}\} \subseteq Col_j$. We prove that any MRO plan $T$ maximizes $\mathsf{Pr}(\bigcup_{a \in A}Col_a=[E])$. \\
\end{theorem}
\begin{proof}
We first consider the simple case of $E\le c$. 

Under this case, if $N\le r_1+R-1$, by Pigeonhole principle, apparently we have $\mathsf{Pr}(\bigcup_{a \in A}Col_a=[E]) =1$ for any MRO plan.

Otherwise $N\le r_1+R-1$, then $|S_1|=r_1$. For any placement plan $T$, the probability of recovering all experts is upper bounded by the probability of recovering expert 1:\\
\begin{align}
\begin{aligned}
    & \mathsf{Pr}(\bigcup_{a \in A}Col_a=[E]) \le \mathsf{Pr}(1\in \bigcup_{a \in A}Col_a)  \\
\end{aligned}
\end{align}

For any placement plan $T$, the probability of recovering expert 1 satisfies:\\
\begin{align}
\label{eq:exp1}
\begin{aligned}
    & \mathsf{Pr}(1\in \bigcup_{a \in A}Col_a) \le 1-\frac{\tbinom{N-r_1}{R} }{\tbinom{N}{R} }  \\
\end{aligned}
\end{align}
For any MRO plan, by definition, we have:\\
\begin{align}
\begin{aligned}
    & \{1,\cdots, E\} \subseteq Col_j, j\in S_1  \\
\end{aligned}
\end{align}
Therefore,\\
\begin{align}
\label{eq:mlowc_eq}
\begin{aligned}
    & \mathsf{Pr}(1\in \bigcup_{a \in A}Col_a) \ge \mathsf{Pr}(\bigcup_{a \in A}Col_a=[E]) \ge 1-\frac{\tbinom{N-r_1}{R} }{\tbinom{N}{R} }  \\
\end{aligned}
\end{align}
Combining Inequality~\ref{eq:exp1} and Inequality~\ref{eq:mlowc_eq}, we have:
for $E\le c$, any MRO plan maximizes $\mathsf{Pr}(\bigcup_{a \in A}Col_a=[E])$ and thus is optimal.

To prove the case of $E>c $, we first define two functions $P_T(\cdot,\cdot,\cdot)$ and $P_s(\cdot,\cdot,\cdot)$. $P_T$ is defined as:
\begin{align}
\label{eq:pt_def}
\begin{aligned}
    & P_T(M,n,r)= \mathsf{Pr}(\bigcup_{a \in A}Col_a \supseteq M)  \\
\end{aligned}
\end{align}
where matrix $T\in \mathbbm{N}^{c\times n}$, $A$ is $r$ columns randomly sampled from $n$ columns
, $M$ is a subset of $[E]$. $P_T$ is used to illustrate the probability of recovering the subset $M$ from a sub-matrix $T$.

 For set $M$, we define $M[j]$ as $j$-th smallest element in set $M$. $P_s$ is defined as:
\begin{align}
\begin{aligned}
    & P_s(M,n,r)= \mathsf{Pr}(r \text{ samples cover the first } \lceil \frac{ |M| }{c} \rceil   \text{ segments of vector } v)  \\
\end{aligned}
\end{align}
where vector $v$ has length $n$, with consecutively $\lceil \frac{ |M| }{c} \rceil   $ segments, the $i$-th segment has length $L_{M,i}=r_{M[1+(i-1)*c]}$, $i=1,\cdots,\lceil \frac{ |M| }{c} \rceil -1$, $L_{M,\lceil \frac{ |M| }{c} \rceil }=\min \{n-\sum_{j=1}^{\lceil \frac{ |M| }{c} \rceil -1} L_{M,j},r_{\lceil \frac{ |M| }{c} \rceil }\}$. $P_s$ is defined to illustrate the recover probability of MRO plans.

We prove the optimality of MRO plan when $E>c$ by mathematical induction. We first have the following assumption:
\begin{assumption}
\label{ass:1}
$\forall m'<E, \forall n', r'$,  $\forall \text{ set }  M', |M'|=m'$,\\
\begin{align}
\label{eq:assumption1}
\begin{aligned}
    & \max_{T} P_T(M',n',r')=P_s(M',n',r')  \\
\end{aligned}
\end{align}
\end{assumption}

We want to prove that for $\forall |M|=E, \forall N, R$, 
\begin{align}
\label{eq:target}
\begin{aligned}
\max_{T} P_T(M,N,R)=P_s(M,N,R) \\
\end{aligned}
\end{align}

Proving Equation~\ref{eq:target} indicates that any MRO plan achieves optimal recover probability across all different $T$.

We first consider the case of $|M|>c$.
First if $R=1, |M|>c$, for $\forall T$, $ P_T(M,N,R)=0, P_s(M,N,R)=0$, the claim trivially satisfies.

When $R>1, |M|>c$, for $\forall T$, we can transform $T$ to $T'$ by reordering the columns to let the columns containing 1 be the first consecutive columns. And $\forall T$ we have:\\
\begin{align}
\begin{aligned}
    &  P_T(M,N,R)=P_{T'}(M,N,R)  \\
\end{aligned}
\end{align}

Let $A'$ as the set of $R$ columns randomly sampled on $T'$, $S_{t}$ be the set of different values of column $t$ of matrix $T'$, $C$ is the largest column ID of $T'$ that contains 1. By conditioning on $t$, we have:\\
\begin{small}
    \begin{align}
\label{eq:PT_expansion}
\begin{aligned}
    &  P_{T'}(M,N,R)=\sum_{t=1}^{C} \textsf{Pr}(\min A'=t) \textsf{Pr}(\bigcup_{a \in A'\setminus \{t\} }Col_a \supseteq M\setminus S_t|\min A'=t)  \\
\end{aligned}
\end{align}
\end{small}

If we consider $T''$ as the sub-table of $T'$ composed of its last $N-t$ rows, we have:\\
\begin{small}
\begin{align}
\begin{aligned}
    &  \textsf{Pr}(\bigcup_{a \in A'\setminus \{t\} }Col_a \supseteq M\setminus S_t|\min A'=t) \le \max_{T''} P_{T''}(M\setminus S_t, N-t, R-1) \\
\end{aligned}
\end{align}
\end{small}

By Assumption~\ref{ass:1}, due to $S_t \neq \emptyset$, we have:\\
\begin{align}
\begin{aligned}
    &  \max_{T''} P_{T''}(M\setminus S_t, N-t, R-1)=P_s(M\setminus S_t,N-t,R-1) \\
\end{aligned}
\end{align}
Recall Equation~\ref{eq:PT_expansion}, we have:\\
\begin{align}
\label{eq:ps_form}
\begin{aligned}
    &  P_{T'}(M,N,R)\le \sum_{t=1}^{r_{M[1]}} \textsf{Pr}(\min A'=t) P_s(M\setminus S_t,N-t,R-1)  \\
\end{aligned}
\end{align}
To upper bound $P_{T'}(M,N,R)$, we have to upper bound $ P_s(M\setminus S_t,N-t,R-1)$.
We first prove the following proposition:
\begin{proposition}
\label{prop:1}
   Denote $\textsf{Min}_c M$ as the smallest $c$ elements of $M$. For $\forall M$, we have:
   \begin{align}
\begin{aligned}
    & \textsf{Min}_c M= \arg\max_{S_t} P_s(M\setminus S_t,N-t,R-1)  \\
\end{aligned}
\end{align}

\end{proposition}

It is apparent that removing elements from the recover target set  results in an increase of $P_s$. Therefore, if $|S_t|< c, \forall s \neq S_t$,\\
\begin{align}
\begin{aligned}
    &  P_s(M\setminus (S_t \cup s),N-t,R-1) \ge  P_s(M\setminus S_t,N-t,R-1)  \\
\end{aligned}
\end{align}
Therefore the set $S_t$ that maximizes $P_s(M\setminus S_t,N-t,R-1)$ must have $c$ cardinality.

Consider $|S_t|=c$. If $S_t$ is not the smallest $c$ elements of $M$, we substitute an element in $S_t$ with a smaller element obtaining $S'_t$, $|S'_t|=c$. By the property of rankings, we have,\\
\begin{align}
\begin{aligned}
    &  L_{M \setminus S'_t,i}\ge L_{M \setminus S_t,i}, \forall i \\
\end{aligned}
\end{align}
Therefore, $\forall S'_t$ obtained by this way, \\
\begin{align}
\begin{aligned}
    &  P_s(M\setminus S'_t,N-t,R-1) \ge  P_s(M\setminus S_t,N-t,R-1)  \\
\end{aligned}
\end{align}
We recursively apply this substitution and obtains $\textsf{Min}_c M$, therefore, for $\forall S_t$, we have:\\
\begin{align}
\begin{aligned}
    &  P_s(M\setminus\textsf{Min}_c M ,N-t,R-1) \ge  P_s(M\setminus S_t,N-t,R-1)  \\
\end{aligned}
\end{align}
Thus finishes the proof of the proposition. This proposition tells us that $S_t=\textsf{Min}_c M$ maximizes $P_s(M\setminus S_t,N-t,R-1)$.

By Equation~\ref{eq:ps_form} and Proposition~\ref{prop:1}, we have,\\
\begin{align}
\label{eq:relax_st_to_minc}
\begin{aligned}
    &  P_{T}(M,N,R)\le \sum_{t=1}^{r_{M[1]}} \textsf{Pr}(\min A'=t) P_s(M\setminus \textsf{Min}_c M,N-t,R-1)  \\
\end{aligned}
\end{align}

For $P_s(M,N,R)$, consider the left most sample should fall on  the first segment, and the other $R-1$ samples should cover the set $M'$, where $M'$ satisfies the $j$-th segment of $M'$ has equal length with the $j+1$-th segment of $M$ for $\forall j$. Therefore  $M'=\{M[1+c],\cdots,M[|M|]\}$.\\
\begin{small}
    \begin{align}
\label{eq:relax_ps}
\begin{aligned}
  P_s(M,N,R)&= \sum_{t=1}^{r_{M[1]}} \textsf{Pr}(\min A'=t) P_s(M' ,N-t,R-1)\\
 &= \sum_{t=1}^{r_{M[1]}} \textsf{Pr}(\min A'=t) P_s( \{M[1+c],\cdots,M[|M|]\},N-t,R-1)  \\
    &= \sum_{t=1}^{r_{M[1]}} \textsf{Pr}(\min A'=t) P_s(M\setminus \textsf{Min}_c M,N-t,R-1)  \\
\end{aligned}
\end{align}
\end{small}

Substituting Equation~\ref{eq:relax_ps} into  Inequality~\ref{eq:relax_st_to_minc}, we have:\\
\begin{align}
\label{eq:upper_bound}
\begin{aligned}
    &  P_{T}(M,N,R)\le P_s(M,N,R)  \\
\end{aligned}
\end{align}

Now we have proven that $P_s$ is an upper bound of $P_T$. Next, we prove that if $T$ is a MRO plan, Inequality~\ref{eq:upper_bound} can actually achieve equal. For $\forall $ MRO plan $T^*$, we have:\\
\begin{align}
\label{eq:mro_ps}
\begin{aligned}
\bigcup_{a \in A}Col_a=[E] \iff A \text{ covers } S_i, \forall i\in \{1,\cdots,\lceil\frac{E}{c} \rceil\}
\end{aligned}
\end{align}

For $\forall $ MRO plan $T^*$, we can reorder the columns so that for each column set $S_i$, all  columns in $S_i$ are  consecutive. We denote the reordered MRO plan as $T'$, and the randomly sampled columns on $T'$ as $A'$.
\begin{align}
\label{eq:reorder_}
\begin{aligned}
&\textsf{Pr}(\bigcup_{a \in A'}Col_a=[E])\\
=&\textsf{Pr}(A' \text{ covers segment with length } |S_i|, \forall i\in \{1,\cdots,\lceil\frac{m}{c} \rceil\})\\
=& P_s(M,N,R)\\
\end{aligned}
\end{align}

Therefore for $T^*$ which is a MRO plan, by the definition of $P_T$ in Equation~\ref{eq:pt_def}, we have:\\
\begin{align}
\label{eq:reorder}
\begin{aligned}
P_{T^*}(M,N,R)= P_s(M,N,R)\\
\end{aligned}
\end{align}

Equation~\ref{eq:reorder} indicates that $\exists \text{MRO plan~} T^*, P_{T^*}(M,N,R)= P_s(M,N,R)$, hence we prove that, under Assumption~\ref{ass:1}, Equation~\ref{eq:target} holds when $E>c$.\\

Assumption~\ref{ass:1} trivially holds due to the optimality of MRO plan when $E\le c$.

By mathematical reduction, for $\forall E, \forall |M|=E, \forall N,R$, we have,\\
\begin{align}
\label{eq:final}
\begin{aligned}
\max_{T} P_T(M,N,R)=P_s(M,N,R) \\
\end{aligned}
\end{align}

Furthermore, for $\forall$ MRO plan $T^*$ we have:\\
\begin{align}
    \begin{aligned}
        P_{T^*}([E],N,R) = \max_T \mathsf{Pr}(\bigcup_{a \in A}Col_a = [E])
    \end{aligned}
\end{align}

\end{proof}

\end{document}